\documentclass[11pt]{article}
\usepackage[T1]{fontenc}
\usepackage{lmodern}
\usepackage[margin=1in]{geometry}
\usepackage{amsmath,amssymb,amsthm,mathtools}
\usepackage{microtype,booktabs,enumitem,array,needspace}
\usepackage[numbers,sort&compress]{natbib}
\usepackage[dvipsnames]{xcolor}
\usepackage[colorlinks=true,linkcolor=MidnightBlue,citecolor=MidnightBlue,urlcolor=MidnightBlue]{hyperref}
\newcommand{\doi}[1]{doi: \href{https://doi.org/#1}{\nolinkurl{#1}}}
\usepackage{fancyhdr}

\usepackage{bbm}
\usepackage{authblk}

\numberwithin{equation}{section}
\newtheorem{theorem}{Theorem}[section]
\newtheorem{proposition}[theorem]{Proposition}
\newtheorem{lemma}[theorem]{Lemma}
\newtheorem{corollary}[theorem]{Corollary}
\theoremstyle{remark}

\newcommand{\RP}{\mathbb{RP}}
\newcommand{\R}{\mathbb R}

\newcommand{\Z}{\mathbb Z}
\newcommand{\E}{\mathbb E}
\newcommand{\SO}{\mathrm{SO}(3)}
\newcommand{\so}{\mathfrak{so}(3)}
\newcommand{\tr}{\operatorname{Tr}}
\newcommand{\supp}{\operatorname{supp}}
\newcommand{\diag}{\operatorname{diag}}
\newcommand{\HS}{\mathrm{HS}}
\newcommand{\op}{\mathrm{op}}
\newcommand{\dd}{\mathrm d}
\newcommand{\av}[1]{\left\langle #1\right\rangle}
\newcommand{\norm}[1]{\left\lVert #1\right\rVert}
\newcommand{\abs}[1]{\left\lvert #1\right\rvert}

\newcommand{\calE}{\mathcal{E}}
\newcommand{\calL}{\mathcal{L}}
\newcommand{\calC}{\mathcal{C}}

\newcommand{\nv}{\mathbbm{n}}
\newcommand{\Q}{\mathsf Q}

\title{Faster-than-power-law decay of order parameter correlations
in the two-dimensional $\RP^2$ model:
a rigorous result by \textit{vibe-proving}}

\author{Masaki Oshikawa}
\affil{Department of Physics, the Ohio State University, Columbus, OH 43210, USA}   

\affil{Institute for Solid State Physics,
the University of Tokyo, Kashiwa, Chiba 277-8581, Japan}

\affil{Kavli Institute for the Physics and Mathematics of the Universe (WPI) 
\\
the University of Tokyo, Kashiwa, Chiba 277-8583, Japan}

\date{October 22, 2026}

\begin{document}
\maketitle
\thispagestyle{plain}

\begin{abstract}
We consider the $\RP^2$ model on the two-dimensional square lattice with nearest-neighbor interactions.
It is obtained by replacing the spin in the standard classical Heisenberg model with
a director, which does not distinguish two opposite directions.
Unlike the standard classical Heisenberg model or the O(3) nonlinear sigma model,
the $\RP^2$ (and its cousin $\RP^3$) model allows point vortex defects classified by $\pi_1(\RP^2)=\Z_2$
in two dimensions.
It is tempting to expect the $\Z_2$ vortices to induce a phase transition analogous to the well-established
Berezinskii-Kosterlitz-Thouless transition in the O(2) model (XY model) in two dimensions.
Indeed, the possibilities of a critical phase or a critical point with power-law correlations
of the order parameter in the $\RP^2$ and $\RP^3$ models have been discussed in the literature.
However, the challenging nature of the problem has prevented a definitive conclusion on the phase diagram so far.
In this paper, we report a rigorous mathematical proof by Artificial Intelligence (ChatGPT 6 Astra) that
the absolute value of the two-point correlation function of order parameters at distance $r$
in the $\RP^2$ and $\RP^3$ models on the two-dimensional square lattice is bounded above by
$C_\beta \exp\{-c_\beta[\log(2+r)]^2\}$ for any finite inverse temperature $\beta\ge0$,
with positive constants $C_\beta$ and $c_\beta$ depending only on $\beta$.
Namely, the correlation function decays faster than any power of $r$.
The proof is formalized and verified in Lean;
the formalization is publicly
available at \href{https://doi.org/10.5281/zenodo.23548847}{doi:10.5281/zenodo.23548847}.
A heuristic argument based on a one-parameter effective theory ansatz at large distances
combined with the rigorous correlation bound suggests (but does not prove) that
the $\Z_2$ vortices are always unbound at any nonzero temperature.
\end{abstract}

\section{Introduction and Main Result}
\label{sec:introduction}

October 6, 2026 turned out to be a historic day in mathematics and mathematical sciences, if not for humanity as a whole.
On that day, OpenAI released more than 700 ``preprints'',
each claiming a significant mathematical proof~\cite{OpenAIRelease2026,OpenAICollection2026}
produced by Artificial Intelligence (AI).
This poses a serious question about the role of humans, or about the relation between humans
and AI, in (even the near) future.
On the other hand, the mathematical prowess of AI certainly gives us unprecedented opportunities.
(See Sec.~\ref{sec:personal} for more personal reflections on this issue by the human author.)

One of the results announced in the OpenAI release concerns the following fundamental
problem in statistical mechanics and field theory.
The classical Heisenberg model
\begin{equation}
      \calE = - J \sum_{\langle x, y \rangle} \mathbbm{n}_x \cdot \mathbbm{n}_y ,
\label{eq:model_O3}
\end{equation}
where the sum is taken over nearest-neighbor pairs of sites $\langle x, y\rangle$ on a lattice and
$\mathbbm{n}_x$ is a three-component unit vector (${\mathbbm{n}_x}^2 = 1$) at site $x$,
is one of the fundamental models in statistical mechanics.
It has a global SO(3) symmetry of rotation of all spins.
For this classical model, the Hohenberg--Mermin--Wagner theorem and its classical counterpart
exclude spontaneous breaking of the continuous rotation symmetry at nonzero temperature
in two dimensions~\cite{Hohenberg1967,MerminWagner1966,Mermin1967,Halperin2019}.

At low temperatures, smooth long-wavelength fluctuations of the Heisenberg model are
expected to be described by a continuum field theory, the O(3) nonlinear sigma model
defined by the Lagrangian density
\begin{equation}
      \calL = \frac{1}{2g} (\partial_\mu \mathbbm{n})^2 ,
\end{equation}
in two dimensions, again subject to the constraint $\mathbbm{n}^2 = 1$.
It can be said that the ``target space'' of this model is the two-dimensional sphere $S^2$.
The perturbative renormalization group analysis of the O(3) nonlinear sigma model in two dimensions gives
the RG equation~\cite{Polyakov1975,Brezin1976}
\begin{equation}
\frac{dg}{d \ell} = \frac{1}{2\pi} g^2 + O(g^3),
      \label{eq:RG_O3}
\end{equation}
where $\ell$ is the logarithm of the length scale, in the weak-coupling (small $g$) limit.
This implies that, even if the ``bare'' coupling $g$ at the microscopic scale is small,
the effective coupling grows at larger length scales.
While the perturbative RG analysis alone cannot determine the eventual behavior at
large length scales, the simplest scenario is that the effective coupling grows to infinity as
the length scale goes to infinity.
Since the running coupling $g$ may be regarded as an effective temperature at the given scale,
this proposed monotonic flow towards $g=\infty$ suggests a disordered phase at every nonzero temperature,
with exponentially decaying correlations. This is a physical inference beyond the controlled weak-coupling calculation.

This is the standard ``textbook'' picture of the classical Heisenberg model
or the O(3) nonlinear sigma model in two dimensions, which has been widely accepted in the physics community.
However, it was not without controversy; there was a proposal of a massless phase with algebraic decay
of correlation functions at low temperatures, as we briefly review in Appendix~\ref{app:historyO3}.
We note that, even though the Hohenberg--Mermin--Wagner theorem excludes an ordered phase with spontaneous symmetry breaking
of the continuous rotation symmetry, it does not exclude a phase transition or a massless phase.
This is clearly seen in the O(2) model (XY model) in two dimensions, which has
a low-temperature phase with algebraic decay of correlation functions and a finite-temperature
Berezinskii--Kosterlitz--Thouless (BKT) transition~\cite{Berezinskii1971,Kosterlitz1973,Frohlich1981}.

The recent OpenAI release gives a Lean-verified proof of exponential spin decay in
the O($n$) model for $n\geq3$ on the square lattice~\cite{OpenAI2026}.
It gives a rigorous basis for the standard picture discussed above,
settling the controversy significantly.

There is an interesting variant of the Heisenberg model or the O(3) nonlinear sigma model,
which we call the $\RP^2$ model.
Here, instead of the two-dimensional sphere $S^2$,
the target space is the real projective space $\RP^2$, which is obtained
from the two-dimensional sphere $S^2$ by identifying antipodal points $\vec{n}$ and $-\vec{n}$.
This corresponds to a classical spin model with ``director'' spins, which have
directions but no orientations; a physical example is a rod-like molecule
which is symmetric between its head and tail.
The ordering of the director spins is called nematic order, and the $\RP^2$ model
is a simple model of nematic liquid crystals.
There are various realizations of such a model, but one standard choice is
\begin{equation}
      \calE = - J_q \sum_{\langle x,y \rangle} \left( \mathbbm{n}_x \cdot \mathbbm{n}_y  \right)^2 ,
\label{eq:model_RP2}
\end{equation}
also with the three-dimensional vector spin $\mathbbm{n}_x$ with the constraint ${\mathbbm{n}_x}^2 = 1$.
This is indeed the two-dimensional version of the Lebwohl--Lasher model~\cite{Lebwohl1972}, which
was introduced as a simple model of nematic liquid crystals. The interaction is invariant under the independent spin inversion
$\mathbbm{n}_x \to - \mathbbm{n}_x$ at each site $x$ ($\Z_2$ gauge symmetry).
The global SO(3) rotation symmetry is still present.

Because of the gauge symmetry, the vector spin $\mathbbm{n}$ cannot be an order parameter.
The standard order parameter which can detect the nematic order is the quadrupolar tensor
\begin{equation}
      Q(\mathbbm{n}) = \mathbbm{n}  \mathbbm{n}^T - \frac{1}{3} \mathbbm{1}_3,
\label{eq:def_Q}
\end{equation}
where $\mathbbm{1}_3$ is the $3 \times 3$ identity matrix.
Its two-point correlation function is
\begin{equation}
     \calC_Q(x, y) =
      \langle \tr{\left[ Q(\mathbbm{n}_x) Q(\mathbbm{n}_y) \right]} \rangle
      = \langle \left( \mathbbm{n}_x \cdot \mathbbm{n}_y \right)^2 - \frac{1}{3} \rangle,
\label{eq:def_CQ}
\end{equation}
which is zero if the spins $\mathbbm{n}_x$ and $\mathbbm{n}_y$ are independently random
(i.e. uniformly distributed on the two-dimensional sphere $S^2$).

Since the round target space $\RP^2$ is locally isometric to $S^2$,
the one-loop RG equation is identical to that of the O(3) model~\eqref{eq:RG_O3};
perturbative sigma-model renormalization depends on the local target geometry~\cite{Friedan1980}.
Again, the simplest scenario is that the effective coupling grows monotonically from an arbitrarily small non-zero value
to infinity at large length scales, and the system is disordered at any non-zero temperature.
However, the different global topology of the target space $\RP^2$ from that of $S^2$
\begin{equation}
      \pi_1(\RP^2) = \Z_2 \neq 0 = \pi_1(S^2)
\end{equation}
allows $\Z_2$ vortices as topological point defects~\cite{ToulouseKleman1976,Mermin1979}.
Given the established BKT transition in the O(2) model driven by $\Z$ vortices at finite temperature,
it is natural to expect a similar transition in the RP$^2$ model~\cite{Solomon1981}.
On the other hand, the critical low-temperature phase of the O(2) model has a solid theoretical foundation,
namely the massless free-boson (Gaussian) field theory.
The ``asymptotically free'' RG flow of the $\RP^2$ model suggests a disordered phase even at very low temperatures;
if that is the case, there is no basis to host a critical point or critical phase even in the
presence of $\Z_2$ vortices.
However, the perturbative RG analysis cannot determine the fate of the RG flow at strong couplings,
and it certainly does not exclude the fascinating possibility of a
$\Z_2$ vortex-induced transition or a critical phase.

In condensed matter physics, the $\Z_2$ vortex-driven transition was proposed for
the classical Heisenberg antiferromagnet on the triangular lattice~\cite{KawamuraMiyashita1984a,KawamuraMiyashita1984b}.
The effective low-energy theory of the fluctuating local $120$-degree spin pattern
is a nonlinear sigma model with target space $\SO\simeq\RP^3$~\cite{DombreRead1989,Caffarel2001}.
Since $\pi_1(\RP^3) = \Z_2$, the $\Z_2$ vortices are allowed as in the $\RP^2$ model.
For the round, isotropic $\RP^3$ sigma model, the local perturbative RG is that of O(4),
suggesting (but not proving) a disordered phase at every nonzero temperature.
The effective theory of the frustrated magnet can have different stiffnesses,
so equality of target spaces alone does not identify it with the round projective model~\cite{Caffarel2001}.
The analogy between the $\RP^2$ and $\RP^3$ models is nevertheless suggestive.

Naturally, there have been numerous numerical studies of the $\RP^2$ (and related) model in two dimensions.
While classical statistical mechanics models in two dimensions are generally numerically tractable,
the large correlation length of the $\RP^2$ model at low temperatures makes it rather difficult
to draw definitive conclusions from numerical data.
The numerical data do suggest unbinding of $\Z_2$ vortices at finite temperatures, at least
in a certain length scale. However, the question is whether the presence of $\Z_2$ vortices leads
to a true phase transition or a critical phase.
There have been various observations and claims about the possible criticality or phase transition in the $\RP^2$ model,
including the presence of a critical phase similar to the low-temperature phase of the O(2) model.
In recent years, with the development of more powerful numerical methods, the overall trend of the numerical
study seemed to favor a crossover at finite length scales rather than a true phase transition or critical phase.
However, the situation was certainly far from being settled~\cite{Kunz1992,Ozeki2019,UedaOshikawa2022,Burgelman2023}.

In this paper, we report a rigorous mathematical proof by AI (ChatGPT 6 Astra)
of the absence of a critical phase
or a critical point with power-law correlation of the order parameter,
in the two-dimensional $\RP^2$ model at any non-zero temperature.
It is proved that the correlation function of the order parameter
(quadrupolar tensor) is bounded as
\begin{equation}
      |\calC_Q(x, y)| \leq C_\beta \exp{\left[ - c_\beta \left( \log{(2 + r)} \right)^2 \right]} ,
\label{eq:main_result_simple}
\end{equation}
where $r = \| x - y \|_\infty$,
for some positive constants $C_\beta$ and $c_\beta$
which depend on the dimensionless inverse temperature $\beta = J_q / T$ (with $k_B=1$).
For more precise description of the theorem, see Sec.~\ref{sec:theorem}.
The right-hand side behaves roughly as $1/r^{c_\beta \log{r}}$ at large $r$, which decays faster than any power-law decay
including possible logarithmic corrections.
While it does not establish the exponential decay of correlation functions at all temperatures,
it rigorously rules out several of the previously proposed scenarios involving
a critical phase or a critical point in the $\RP^2$ model with power-law correlations of the order parameter.
A similar bound was also obtained for the $\RP^3$ model, with similar implications.
We do not have a rigorous statement on the binding of $\Z_2$ vortices.
However, based on heuristic physical arguments, the above bound suggests that
the $\Z_2$ vortices are always unbound at any non-zero temperature,
as we will discuss in Sec.~\ref{sec:consequences}.

Nevertheless, the present result still leaves the possibility of a phase transition which does not
involve power-law decay of the correlation function of the quadrupolar order parameter, such as a first-order transition.
It also allows the possibility of an unconventional critical point or critical phase
with ``super-polynomial'' decay of correlation functions satisfying the bound.
To our best knowledge, no such scenario has been proposed for the $\RP^2$ model in the literature.
On the other hand, such superpolynomially decaying correlation functions have been discussed
in several contexts, such as Cardy--Ostlund model with quenched disorder
in two dimensions~\cite{CardyOstlund1982,Ristivojevic2012Superrough,Perret2012Superrough,LeDoussal2013Correlations},
and the sliding columnar phase in three dimensions~\cite{Golubovic2000Sliding}.
(See Ref.~\cite{OHern1998Nonlinear,OHern1999Sliding} for the sliding columnar phase.)
However, these are quite different from the present $\RP^2$ model, and
at this point we see no reason to expect a similar behavior in our case.

On the other hand, it is still natural to expect
that the correlation functions decay exponentially at any non-zero temperature,
although it cannot be inferred from the present result.
Proving (or disproving) this expectation, and clarification of the entire phase diagram, are important open problems
for the future.

We also note that the present result is specific to the concrete model,
and does not rule out the possibility of a conventional critical point or conventional critical phase
in modified models with the same target space and symmetries.

Even with various limitations, the present result is a significant step forward
in one of the long-standing controversies in a fundamental model in statistical mechanics and field theory
involving topological defects.
We hope that this result will stimulate further studies on the $\RP^2$ model and related models, 
and also further applications of AI-assisted (or produced) rigorous results on
problems of significant physical interest.

This paper is organized as follows.
In Sec.~\ref{sec:historyRP2}, we give a historical overview of the studies on the $\RP^2$ model
(and related models with $\Z_2$ vortices) in two dimensions.
An overview of the controversies on the ``parent'' model, the O(3) model in two dimensions,
is given in Appendix~\ref{app:historyO3}.
In Sec.~\ref{sec:consequences}, we discuss the physical consequences of the main theorem
and argue for (but not prove) the absence of $\Z_2$ vortex unbinding transition
in two dimensions at finite temperatures.
In Sec.~\ref{sec:theorem}, we present the main theorem in a precise manner,
together with a brief outline of the proof and its predecessors.
Statements on the formalized proof in Lean are given in Sec.~\ref{sec:formalization}.
In Sec.~\ref{sec:personal}, personal background and reflections by the human author
are given.
Methodologies including the statements on the AI use are given in Sec.~\ref{sec:method}.
The exposition of the proof is given in Appendices~\ref{app:input}--\ref{app:scales}.

\section{Historical overview of studies on the \texorpdfstring{$\RP^2$}{RP2} model in two dimensions}
\label{sec:historyRP2}

Here we briefly review the controversies on the $\RP^2$ model
(and related models) in two dimensions.

\subsection{Topological motivation and the early vortex proposal}
\label{sec:historyRP2:origin}

The local geometry of $\RP^2$ is that of $S^2$, but
$\pi_1(\RP^2)=\Z_2$ permits point defects absent from the $S^2$ model
\cite{ToulouseKleman1976,Mermin1979}. These defects, often called
$\Z_2$ vortices or half-disclinations, motivated an analogy with the
Berezinskii--Kosterlitz--Thouless (BKT) transition. An early explicit
proposal was Solomon's 1981 study of two-dimensional non-Abelian spin
models~\cite{Solomon1981}. For the biquadratic projective interaction,
an energy--entropy argument and a strong-coupling expansion suggested
a vortex-driven transition. The paper also recognized that non-Abelian
spin fluctuations could alter the vortex free energy and reduce the
transition temperature to zero. Thus the existence of vortices did not
by itself establish a positive-temperature transition.

A nonlinear sigma model describes smooth, long-wavelength director
fluctuations at low temperature; its local weak-coupling expansion is
closely related to that of the O(3) model.
Taken naively, it suggests that the $\RP^2$ model is disordered at any nonzero temperature
similarly to the O(3) model, although the latter model has also been
a subject of controversy as reviewed in Appendix~\ref{app:historyO3}.
In any case, the eventual fate of the model in the thermodynamic limit
is determined by the RG flow in the strong-coupling regime
which is not accessible to the perturbative expansion.
Thus the quasi-long-range order with a power-law decay of the quadrupolar correlation function
cannot be ruled out by the perturbative RG analysis alone.

Tensor-field effective-action analyses subsequently identified
additional interactions and suggested departures from the naive O(3)
description~\cite{Yamamoto1992,Yamamoto1993}. These arguments enlarged
the possible infrared scenarios without establishing
algebraic quasi long-range order in the lattice model.

\subsection{Evidence interpreted as a critical point or critical phase}
\label{sec:historyRP2:critical}

Early simulations associated changes in the apparent mass gap and
susceptibility with vortex liberation, while also allowing a very small
nonzero gap or a nearby complex singularity as alternatives
\cite{Fukugita1982}. The influential Kunz--Zumbach studies reported
strong evidence for a topological transition accompanied by divergent
orientational correlation length and susceptibility
\cite{Kunz1991,Kunz1992}. Their correlation-length data admitted both
ordinary power-law and KT-type singular fits, illustrating the difficulty
of identifying an asymptotic transition from a limited range of scales.

Later studies sharpened, but did not unify, the critical interpretation.
Conformal-mapping calculations supported a BKT-like transition and
algebraic low-temperature order~\cite{Farinas2003}. Standard finite-size
scaling was instead used to argue for a critical point with power-law
correlations and also a critical phase at lower temperatures
in Ref.~\cite{Mondal2003}. Cluster simulations subsequently found
indications of a line of critical points, while acknowledging unresolved
diagnostics~\cite{Shabnam2016}. Nonequilibrium-relaxation calculations
favored KT-type dynamical scaling and estimated critical exponents at
and below the proposed transition~\cite{Ozeki2019}.

The present result excludes the asymptotic power-law decay of
quadrupolar correlations associated with these critical interpretations, within
the model and class of states covered by the theorem. It does not
contradict the reported finite-distance or finite-time observations
themselves: apparent scaling over accessible windows may precede a
crossover. In particular, numerical evidence for slow relaxation alone
does not determine the equilibrium large-distance correlation.

\subsection{Arguments and evidence against algebraic criticality}
\label{sec:historyRP2:noncritical}

Probability distributions, hyperscaling tests, Binder cumulants and
stiffness measurements were used to question a low-temperature line of
critical points and to favor a crossover interpretation
\cite{Paredes2008,Farinas2010}. A comparison of correlation ratios in
the XY, Heisenberg and projective models likewise favored
pseudocritical behavior for $\RP^2$~\cite{Tomita2014}.
More recently, tensor-network renormalization found no convincing
finite-temperature conformal spectrum and suggested a single disordered
phase with a strong crossover~\cite{UedaOshikawa2022}.
Symmetry-preserving matrix-product-state calculations also found
scaling incompatible with true criticality and consistent with a
pseudocritical regime below a large correlation length
\cite{Burgelman2023}. These approaches provide complementary evidence,
although finite system sizes, observation times or tensor bond
dimensions still limit the accessible scales.

Scale-invariant scattering supplied an independent analytical argument
against nematic quasi-long-range order for three director components
\cite{Delfino2021,Diouane2021}. Subsequent work clarified subtleties in
assigning universality classes to the scattering solutions
\cite{Diouane2023}. These continuum analyses do not provide the uniform
lattice correlation estimate proved here. Conversely, our estimate
supports their conclusion about algebraic order without establishing
every stronger claim of a single analytic phase or exponential decay.

The zero-temperature universality question is separate. Finite-size
scaling suggested projective behavior distinct from O($N$)
\cite{Caracciolo1993,Bonati2020}, whereas weak-coupling results for
constraint actions and defect-based arguments supported a common
continuum limit~\cite{Hasenbusch1996,Niedermayer1996}.
Neither proposed universality class entails a positive-temperature
critical phase. In particular, the defect-suppression comparison in
Ref.~\cite{Niedermayer1996} uses an unproved monotonicity assumption
for the tensor correlation length; decay in the O(3) model alone does
not supply that comparison.

In Refs.~\cite{Latha2018,Latha2026}, a low-temperature phase without
the power-law decay of quadrupolar correlations was proposed.
However, the reported exponential relaxation occurs towards an apparent nonzero background in the full quadrupolar correlation function.
Interpreted as thermodynamic nematic long-range order, this background conflicts with the Hohenberg--Mermin--Wagner theorem
(and the present result).
If the finite-size data correctly represent equilibrium correlations,
the apparent plateau must disappear in the appropriate large-size and large-distance limits.
Consequently, the fitted relaxation length would only
describes short-range physics, and
does not by itself establish a thermodynamic phase structure.

\subsection{Other proposed transitions and the limits of the comparison}
\label{sec:historyRP2:other}

Absence of algebraic quadrupolar order does not establish analyticity
of the free energy or exclude a transition detected by another
observable. 

An example of a proposed transition without algebraic
spin order comes from frustrated magnets. Kawamura and Miyashita
proposed $\Z_2$-vortex unbinding in the triangular-lattice Heisenberg
antiferromagnet~\cite{KawamuraMiyashita1984a,KawamuraMiyashita1984b}.
The slowly varying local $120$-degree structure defines an oriented
triad, with order-parameter space $\SO\simeq\RP^3$, and admits a
long-wavelength sigma-model description~\cite{DombreRead1989}.
Subsequent spin-stiffness, vorticity and correlation studies examined
the competition between spin waves and defects
\cite{SouthernYoung1993,SouthernXu1995,Wintel1995,Caffarel2001}.
The later two-length-scale scenario allows a diverging vortex scale
while the spin correlation length remains finite, even below the
transition; the resulting state was called a \emph{spin-gel}
\cite{Kawamura2010,Kawamura2011}. Spin transport was proposed as a
probe of the vortex transition~\cite{Aoyama2020}. Such a separation
of scales would not be excluded by a bound on spin correlations alone.

The target-space correspondence does not identify the triangular
antiferromagnet with the isotropic biquadratic $\RP^3$ lattice model:
the microscopic interactions and effective stiffnesses must also be
specified~\cite{Caffarel2001}. A related large-scale study did simulate
the isotropic square-lattice biquadratic $\RP^3$ Hamiltonian, using it
as an effective model, and obtained evidence compatible with a vortex
transition while finding orientational lengths beyond its largest
size $L=16384$~\cite{Okubo2023}. This Hamiltonian matches the
projective interaction in our extension. The theorem nevertheless
concerns its specified free-boundary measures and local limits, and
does not settle the proposed vortex transition or justify transferring
the result to the frustrated magnet.

Modified Hamiltonians require a further distinction. Introducing
independent $\Z_2$ gauge links and a plaquette coupling allows the
vortex fugacity to be tuned~\cite{SolomonStavansDomany1982}.
Monte Carlo renormalization-group work on this enlarged model found
a first-order vorticity line ending at a critical point and proposed
nearby crossover flows as an explanation of apparent projective
scaling~\cite{Catterall1998}. 
Adding a defect penalty likewise changes the Gibbs measure and,
in planar (two-component) nematic models, the transition itself~\cite{Dutta2004}.
These studies illuminate possible mechanisms,
but their transitions are not claims about the models studied in this paper.

Thus three classes should be distinguished: algebraic quadrupolar
correlations (quasi long-range order) in the specified model,
which is excluded in the states covered
by our theorem; thermodynamic or topological transitions with
sufficiently rapidly decaying full orientational correlations, which
the bound alone does not exclude; and transitions in different
Hamiltonians, to which it does not automatically apply. Compatibility
is not evidence that any remaining scenario occurs. The estimate also
permits superpolynomial but nonexponential decay; it neither identifies
such a phase nor establishes uniqueness of the Gibbs state.

\section{Heuristic discussion of the consequences of the main result}
\label{sec:consequences}

The main result of the present report
is a rigorous upper bound~\eqref{eq:main_result_simple} on the quadrupolar correlation
function $\calC_Q(r)$ at large separation $r$ (see also Theorem~\ref{thm:main} in Sec.~\ref{sec:theorem}).
By itself, it makes no statement
about $\Z_2$ vortices. In this section, we combine the bound with a simple
effective theory and an assumption of monotone renormalization-group (RG)
flow, and argue heuristically that the $\Z_2$ vortices are unbound at every
nonzero temperature. The argument uses only the inequality
\eqref{eq:main_result_simple}, not a hypothesis on the actual asymptotic
form of $C_Q(r)$. It therefore applies equally if the true decay is
exponential, which remains a natural possibility not established by the
present result.

Let us assume the simple physical picture that both spin fluctuations and
vortex textures can be described by the same running, isotropic
stiffness in the simple effective theory
\begin{equation}
 F_{\mathrm{eff}}[\mathbbm{n}]/T
 =\frac{1}{2g(\ell)}\int\dd^2x\,|\nabla\mathbbm{n}|^2,
 \qquad g(\ell)=T/\rho_s(\ell).
 \label{eq:elastic}
\end{equation}
Here $g(\ell)$ is the running coupling in the effective theory at length
$a e^\ell$, where $a$ is the lattice spacing, and $T$ is the physical
temperature (with $k_B=1$). We use leading spin-wave and dilute-defect
formulae with this running coupling, neglecting additional defect
renormalizations or possible extra terms generated by RG.
We note that we consider the dependence of the effective theory
on the logarithmic length scale $\ell$, while $T$ is fixed to the given
physical temperature.
Real-projective sigma models and their $\Z_2$ defects are discussed,
for example, in Ref.~\cite{Nahum2013}.

The leading renormalization of the quadrupolar order parameter is obtained
by integrating over a thin momentum shell in a local
tangent frame. Each tangent component $\pi_\alpha$ has variance
\begin{equation}
 \langle\pi_\alpha^2\rangle_{\mathrm{shell}}
 = g \frac{\dd\ell}{2\pi }=:v,\qquad \alpha=1,2.
\end{equation}
Expanding the director about its slow component gives
$\langle\Q\rangle_{\mathrm{shell}}=(1-3v)\Q_{\mathrm{slow}}$ to this
order. Applying this result to both operators in $\calC_Q$, we find
\begin{equation}
 -\frac{\dd\log \calC_Q}{\dd\ell}
 =\frac{3g(\ell)}{\pi}+O(g^2),
 \label{eq:eta}
\end{equation}
where $\calC_Q$ is now regarded as a function of $\ell=\log(r/a)$.
While more analysis is needed to determine the right-hand side in
the strong-coupling region, we take a simple ansatz that the leading-order
result in Eq.~\eqref{eq:eta} is valid for general $g(\ell)$.

\subsection{Lower bound on the running coupling}

Adopting the extrapolated leading-order relation in Eq.~\eqref{eq:eta} and
integrating it from a reference scale $\ell_0>0$ gives
\begin{equation}
 \calC_Q(\ell)= \calC_Q(\ell_0)
 \exp\!\left[-\frac3\pi\int_{\ell_0}^{\ell}g(\ell')\,\dd\ell'\right].
 \label{eq:integrated}
\end{equation}
Within the effective theory, $\calC_Q(\ell_0)>0$ for the nematic interaction
$J_q>0$. Since $[\log(2+r)]^2\ge\ell^2$ (in units with $a=1$), comparison
with the rigorous bound \eqref{eq:main_result_simple} yields
\begin{equation}
 \int_{\ell_0}^{\ell}g(\ell')\,\dd\ell'
 \ge\frac\pi3\left(c_\beta\ell^2-K_\beta\right),
 \qquad
 K_\beta=\max\!\left\{0,\log\frac{C_\beta}{C_Q(\ell_0)}\right\}.
 \label{eq:integralbound}
\end{equation}

This inequality constrains only the integral of $g$.
In principle, a coupling that is
very large on most scales could remain small on a long but sparse set
of scales, and the vortex free energy computed below, which involves
$\int\dd\ell/g$, could accumulate on that set. We exclude such a possibility
by the following natural, but unproven assumption (in addition to what we have already
assumed concerning the effective theory~\eqref{eq:elastic}):
\begin{quote}
\textit{Monotone flow.} For $\ell\ge\ell_0$, the running coupling $g(\ell)$
is nondecreasing; equivalently, the effective stiffness $\rho_s(\ell)$ does
not increase under coarse graining.
\end{quote}
This is the behavior suggested by the positive one-loop beta function
\eqref{eq:RG_O3}, which is shared by the $\RP^2$ model. Vortex screening,
if present, is expected to reduce the stiffness further and thus
reinforce the assumption. It is nevertheless not established beyond weak
coupling.

Under this assumption, $\int_{\ell_0}^{\ell}g\,\dd\ell'\le(\ell-\ell_0)\,g(\ell)
\le\ell\,g(\ell)$, and Eq.~\eqref{eq:integralbound} gives
$g(\ell)\ge(\pi/3)(c_\beta\ell-K_\beta/\ell)$. Hence
\begin{equation}
 g(\ell)\ge\kappa_\beta\,\ell,
 \qquad \kappa_\beta=\frac{\pi c_\beta}{6},
 \qquad \text{for}\ \ell\ge\ell_1
 :=\max\!\left\{\ell_0,\sqrt{2K_\beta/c_\beta}\right\}.
 \label{eq:running}
\end{equation}
The coupling must thus grow at least linearly in $\ell$, and
$\rho_s(\ell)\le T/(\kappa_\beta\ell)$ flows to zero.
For comparison, a correlation function saturating the bound,
$C_Q\sim C_0\exp(-A\ell^2)$, would give $g(\ell)\simeq2\pi A\ell/3$;
the monotonicity argument loses only a numerical factor.
A weaker assumption also suffices: if $g(\ell')\le M g(\ell)$ for all
$\ell_0\le\ell'\le\ell$ with a fixed $M\ge1$, Eq.~\eqref{eq:running} holds
with $\kappa_\beta$ replaced by $\kappa_\beta/M$.

The role of the theorem in this argument is to exclude a flow towards
a finite fixed point $g(\ell)\to g_*<\infty$. Within the same ansatz,
such a flow would give a power law $C_Q\sim r^{-3g_*/\pi}$, and the
standard vortex argument below would bind the $\Z_2$ vortices when
$g_*<\pi/8$. This BKT-like scenario is precisely what the bound rules out.

\subsection{Free energy of \texorpdfstring{$\Z_2$}{Z2} vortices}

Now let us consider the free energy in the presence of $\Z_2$ vortices,
assuming the effective theory in Eq.~\eqref{eq:elastic} with
the running coupling satisfying Eq.~\eqref{eq:running}.
For a single vortex, take a disk of radius $L$ whose boundary enforces
the nontrivial element of $\pi_1(\RP^2)=\Z_2$, and initially pin the
core at its center. The minimal half-turn texture has the representative
\begin{equation}
 \nv(r,\theta)=\left(\cos\frac\theta2,\sin\frac\theta2,0\right),
 \qquad |\partial_\theta\nv|^2=\frac14.
 \label{eq:vortex}
\end{equation}
Although its lift changes sign after one circuit, the director is
single-valued. Inserting this texture into Eq.~\eqref{eq:elastic} and
integrating one annulus yields
\begin{equation}
 \frac{\dd F_v^{\mathrm{pin}}}{T}
 \simeq\frac{1}{2 g(\ell)}
       \int_0^{2\pi}\frac14\,\dd\theta\,\dd\ell
 =\frac{\pi}{4g(\ell)}\,\dd\ell.
 \label{eq:shell}
\end{equation}
The notation denotes the constrained excess free energy within this
elastic approximation, in a vortex-free background.
Core physics and shorter scales are absorbed at the reference
length $L_1=ae^{\ell_1}$. Equation~\eqref{eq:running} then gives
\begin{equation}
 F_v^{\mathrm{pin}}(L)-F_v^{\mathrm{pin}}(L_1)
 \lesssim\frac{\pi T}{4\kappa_\beta}
 \log\frac{\log(L/a)}{\ell_1}
 =\frac{3T}{2c_\beta}\log\frac{\log(L/a)}{\ell_1}.
 \label{eq:single}
\end{equation}
The cost of an isolated vortex therefore grows at most as
$\log\log L$, instead of the $\log L$ of a vortex with finite stiffness.

Following the standard argument for the BKT transition~\cite{Kosterlitz1973},
let us take the positional entropy into account.
Allowing the core of a single $\Z_2$ vortex to move through a bulk region
of area proportional to $L^2$ contributes $2\log(L/a)$ to its entropy.
Up to a size-independent matching constant $F_*$, Eq.~\eqref{eq:single}
gives
\begin{equation}
 F_v^{\mathrm{mobile}}(L)
 \lesssim F_*+\frac{3T}{2c_\beta}\log\log(L/a)-2T\log(L/a).
 \label{eq:mobile}
\end{equation}
For every fixed $T>0$ and finite core free energy, the positional term
eventually dominates for large enough system size $L$.
Although this is just an analysis of a single vortex,
it suggests that the $\Z_2$ vortices proliferate at any nonzero temperature,
unlike in the standard BKT transition.

The same instability appears directly for a neutral pair.
Fix two cores at separation $R$ and allow their orientations to
relax in the topologically trivial total sector. Both are the same
nontrivial $\Z_2$ species. An annihilating pair can nevertheless be
represented in a common plane by windings $+1/2$ and $-1/2$; the signs
are not distinct topological charges of $\RP^2$. There are two vortex
contributions on scales smaller than $R$, while the far-field winding
cancels. The leading elastic interaction, by the usual vortex-pair
calculation~\cite{Kosterlitz1973}, obeys
\begin{equation}
 \frac{\dd F_2}{\dd\log R}
 \simeq\frac{\pi T}{2 g(\log(R/a))}
 \lesssim\frac{\pi T}{2\kappa_\beta\log(R/a)}.
\end{equation}
With the matching separation $R_1=L_1$, this gives
\begin{align}
 F_2(R)-F_2(R_1)
 &\lesssim\frac{3T}{c_\beta}
 \log\frac{\log(R/a)}{\ell_1}\,,
 \label{eq:pair}\\[3pt]
 0<-f_R=\frac{\dd F_2}{\dd R}
 &\lesssim\frac{3T}{c_\beta\,R\log(R/a)}.
 \label{eq:force}
\end{align}
The force $f_R$ is attractive, but weaker than $1/(R\log R)$ at large
separation. Here $F_2$ is the free energy for
specified core positions, with the common bulk contribution subtracted;
it does not include the number of possible relative positions. These
formulae concern relaxed defects in the dilute approximation, not all
possible constrained core orientations or a screened vortex plasma.

With one core fixed, the relative-position measure is $2\pi R\,\dd R$.
Setting $p=3/c_\beta$, Eq.~\eqref{eq:pair} gives, up to constant factors,
\begin{equation}
 Z_{\mathrm{rel}}(L)\ \gtrsim
 \int_{R_1}^{L}\frac{R\,\dd R}{a^2}
       [\log(R/a)]^{-p}
 \sim\frac{(L/a)^2}{2[\log(L/a)]^p}.
 \label{eq:pairentropy}
\end{equation}
The relative-separation distribution is therefore not normalizable
at infinite volume within this two-defect approximation. Equivalently,
the free energy per logarithmic separation interval is
$F_2(R)-2T\log(R/a)$ plus a constant. For intervals of fixed width
$\dd R$, the geometrical term is instead $-T\log(R/a)$.
This also suggests that $\Z_2$ vortices do not form bound pairs
and rather proliferate as a vortex plasma at any nonzero temperature.

This conclusion covers every correlation function compatible with
the bound~\eqref{eq:main_result_simple}, including exponential decay,
for which $g(\ell)$ grows
even faster. A proliferated vortex plasma would, in turn, be expected to
screen the quadrupolar correlations and produce exponential decay;
a correlation function that actually saturates the log-square bound
therefore does not appear self-consistent within this picture.
Thus the main result of the present paper, the rigorous bound on the
quadrupolar correlation function, suggests the absence of a phase
transition associated with $\Z_2$ vortex unbinding at any nonzero temperature.

It should be emphasized, however, that the discussion in this section is
not only mathematically non-rigorous, but also relies on uncontrolled
physical assumptions: the simple effective theory~\eqref{eq:elastic}
describing both the spin fluctuations and the $\Z_2$ vortex textures with
a single running coupling $g(\ell)$; the extension of the leading-order
relation~\eqref{eq:eta} to strong coupling; and the monotonicity of
$g(\ell)$. The claims in the present section are therefore not established.
For example, distinct spin and vortex length scales have been proposed in
frustrated Heisenberg magnets, a different model~\cite{Kawamura2011}.
If such a separation of scales occurs in the present model, the
discussion in this section could be invalidated.

\section{The main theorems and overview of the proof}
\label{sec:theorem}

\subsection{Statement of the theorem}

Let $G=(V,E)$ be any finite subgraph of the nearest-neighbor square lattice.
For $|b_e|\le\beta<\infty$, where $\beta\ge0$, consider the free-boundary measure
\begin{equation}\label{eq:model}
 \dd\mu_{G,b}(n)=\frac1{Z_{G,b}}
 \exp\!\left(\sum_{\{x,y\}\in E}b_{xy}(\mathbbm{n}_x\cdot \mathbbm{n}_y)^2\right)
 \prod_{x\in V}\dd\omega(\mathbbm{n}_x),\qquad \mathbbm{n}_x\in S^2,
\end{equation}
where $\omega$ is the uniform probability measure. 
This is a generalized version of the model~\eqref{eq:model_RP2} introduced in the Introduction,
with bond-dependent couplings $b_{xy}$, which can be positive or negative.
A convention with interaction
$P_2(\mathbbm{n}_x\cdot \mathbbm{n}_y)$, where
$P_2$ is the second-order Legendre polynomial, is obtained by rescaling the couplings by $3/2$.

We study the two-point function of the quadrupolar tensor~\eqref{eq:def_Q},
\begin{equation}
      \calC_Q^{G,b}(x,y)=\langle \tr(Q(\mathbbm{n}_x)Q(\mathbbm{n}_y))\rangle_{\mu_{G,b}} ,
\end{equation}
which is the generalization of Eq.~\eqref{eq:def_CQ} to the inhomogeneous model~\eqref{eq:model}.
The two-point function defined as $\langle P_2(\mathbbm{n}_x\cdot \mathbbm{n}_y) \rangle$
equals $(3/2)\calC_Q^{G,b}(x,y)$;
its bounds therefore have the same rates and a prefactor multiplied by $3/2$.
We use $r=\norm{x-y}_\infty$.  Rotation invariance gives
$\av{Q(\mathbbm{n}_x)}=0$.  Also
\begin{equation}\label{eq:Qfacts}
 Q(R \mathbbm{n})=RQ(\mathbbm{n})R^T,\qquad
 \frac1{12}\sum_{R\in\mathcal T}RQ(\mathbbm{n})R^T=0,\qquad
 \norm{Q(\mathbbm{n})}_{\mathrm F}^2=\tfrac23,
\end{equation}
where $\norm{Q}_{\mathrm F} = \sqrt{\tr{Q^TQ}}$ is the Frobenius norm. 
Here $\mathcal T$ is the group of twelve orientation-preserving signed
cyclic coordinate permutations: compose a cyclic permutation with one of
the four diagonal sign matrices of determinant $+1$. The sign average
kills off-diagonal entries of $\mathbbm{n}\mathbbm{n}^T$, and the cyclic average makes its
three diagonal entries equal to $1/3$. This proves the cancellation in
\eqref{eq:Qfacts} without Haar integration. This is a standard finite-group
averaging identity: the second moment of each tetrahedral orbit is
isotropic. In the terminology of spherical designs, these orbits are
spherical $2$-designs; see the general framework of
Delsarte, Goethals and Seidel~\cite{Delsarte1977}. We use the explicit
calculation above, rather than a design-existence theorem.
In particular, $\abs{\tr(Q(\mathbbm{n})Q(\mathbbm{m}))}\le2/3$ for any unit vectors $\mathbbm{n},\mathbbm{m}$.

\Needspace{13\baselineskip}
\begin{theorem}[Lean-verified log-square decay]
\label{thm:main}
For every finite $\beta\ge0$, there are constants $C_\beta,c_\beta>0$
such that, for every finite subgraph $G$ of the nearest-neighbor square-lattice graph, every
$b\in[-\beta,\beta]^E$ and every $x,y\in V$,
\begin{equation}
\label{eq:logsquare}
 |\calC_Q^{G,b}(x,y)|\le C_\beta
       \exp\{-c_\beta[\log(2+\|x-y\|_\infty)]^2\}.
\end{equation}
The constants are selected before the graph, couplings and observation
sites. The estimate holds at every distance, including coincident sites,
and passes with the same constants to the free-boundary local limits
specified in Corollary~\ref{cor:limits}.
\end{theorem}
The all-temperature envelope \eqref{eq:logsquare} is faster than every
power, but weaker than $\exp(-c r^\alpha)$ for every fixed $\alpha>0$.

\begin{theorem}[Lean-verified high-temperature exponential decay]\label{thm:highT}
If $0\le\beta<\log(3/2)$, set
\begin{equation}\label{eq:highmass}
 m_\beta=-\log\!\left(\max\left\{\tfrac12,
                              3(1-e^{-\beta})\right\}\right)>0.
\end{equation}
For the same graphs, couplings and sites,
\begin{equation}\label{eq:highexp}
 |\calC_Q^{G,b}(x,y)|\le\frac89 e^{-m_\beta\|x-y\|_\infty}.
\end{equation}
The mass and prefactor are unchanged in the stated free-boundary local
limits. 
\end{theorem}

The stronger, exponential-decay bound is thus established,
but only at sufficiently high temperatures.

\subsection{The \texorpdfstring{$\RP^3$}{RP3} model}
\label{sec:theoremRP3}

The same results hold for the projective model with four-component
spins. Replace $\mathbbm{n}_x\in S^2$ by a unit vector
$\mathbbm{q}_x\in S^3\subset\R^4$ and consider, on the same graphs and
with the same couplings, the free-boundary measure
\begin{equation}\label{eq:modelRP3}
 \dd\mu^{(4)}_{G,b}(\mathbbm{q})=\frac1{Z^{(4)}_{G,b}}
 \exp\!\left(\sum_{\{x,y\}\in E}b_{xy}(\mathbbm{q}_x\cdot\mathbbm{q}_y)^2\right)
 \prod_{x\in V}\dd\omega_3(\mathbbm{q}_x),
\end{equation}
where $\omega_3$ is the uniform probability measure on $S^3$. The weight is
invariant under $\mathbbm{q}_x\to-\mathbbm{q}_x$ at each site, so the spins
are effectively elements of $S^3/\{\pm1\}=\RP^3$. The quadrupolar tensor
and its correlation function are
\begin{equation}\label{eq:defQRP3}
 Q_4(\mathbbm{q})=\mathbbm{q}\mathbbm{q}^T-\tfrac14\mathbbm{1}_4,\qquad
 \calC_{(4)}^{G,b}(x,y)=\av{\tr[Q_4(\mathbbm{q}_x)Q_4(\mathbbm{q}_y)]}
 =\av{(\mathbbm{q}_x\cdot\mathbbm{q}_y)^2-\tfrac14},
\end{equation}
which satisfies $|\tr[Q_4(\mathbbm{q})Q_4(\mathbbm{q}')]|\le3/4$.

\begin{theorem}[Lean-verified decay in the $\RP^3$ model]\label{thm:RP3}
For every finite $\beta\ge0$, there are constants
$C^{(4)}_\beta,c^{(4)}_\beta>0$ such that, for every finite subgraph $G$ of
the nearest-neighbor square-lattice graph, every $b\in[-\beta,\beta]^E$ and
every $x,y\in V$,
\begin{equation}\label{eq:logsquareRP3}
 |\calC_{(4)}^{G,b}(x,y)|\le C^{(4)}_\beta
       \exp\{-c^{(4)}_\beta[\log(2+\|x-y\|_\infty)]^2\}.
\end{equation}
If moreover $0\le\beta<\log(3/2)$, then
\begin{equation}\label{eq:highexpRP3}
 |\calC_{(4)}^{G,b}(x,y)|\le e^{-m_\beta\|x-y\|_\infty},
\end{equation}
with the same mass $m_\beta$ as in \eqref{eq:highmass}. Both bounds pass
with the same constants to the free-boundary local limits of these
measures, defined as in Corollary~\ref{cor:limits}.
\end{theorem}

\subsection{Proof strategy and historical context}
\label{sec:ancestry}

The main bound comes from accumulating fixed contractions of the
quadrupolar correlation across nested square annuli. A fourth-order
response inequality for noncommuting spin rotations
(Appendix~\ref{app:input}), combined with Fourier test functions whose
spatial support is retained in the estimates
(Appendix~\ref{app:fourier}), controls the change in an annular partition
kernel under relative rotations of its boundaries. At sufficiently large
radius $R$, annuli of width proportional to $R/\log R$, with a
$\beta$-dependent proportionality constant, have relative comparison error
at most $1/2$, uniformly in the spins pinned on both boundaries
(Appendix~\ref{app:annulus}). Each kernel can therefore be written as an
equally weighted mixture of its rotational average and a nonnegative
remainder. An averaged annulus separating the observation sites
annihilates the quadrupolar numerator by \eqref{eq:Qfacts}, while replacing
all selected kernels by their remainders preserves the partition
function exactly. Thus each annulus contributes a factor $1/2$ to the
correlation bound, without an independence assumption
(Appendix~\ref{app:averaging}). Order $\log R$ disjoint annuli fit within
each dyadic band of radii; summing over bands up to the separation
$r=\|x-y\|_\infty$ gives order $(\log r)^2$ contractions. Their product
yields \eqref{eq:logsquare}, with the bounded-distance regime absorbed
into the prefactor (Appendix~\ref{app:thinpacking}). No restriction on
projective defects is imposed. Appendix~\ref{app:scales} passes the bound
to the specified local limits and derives weighted summability and
Fourier smoothness.

The proof builds on established rotation and averaging methods
\cite{DobrushinShlosman1975,Shlosman1978}. Its direct analytic input is the
fourth-order noncommuting-rotation inequality and multiscale Fourier
construction of Ref.~\cite[Sections 2--3]{OpenAI2026}, adapted here to the
squared interaction, signed couplings, and arbitrary pinned arrays.
The change-of-variables identities belong to the Ward-identity framework
\cite{Driessler1979,AizenmanSimon1980}; products of contraction factors
from nested rotation layers occur in
Ref.~\cite[Section 2.1]{Ioffe2002}. The positive kernel decomposition is
an elementary splitting construction, related in form to
Ref.~\cite{Nummelin1978}. Complex-rotation arguments provide another
route to algebraic bounds \cite{McBryan1977,Gagnebin2014}; the present
proof uses real changes of spin variables, with complex multilinear
test functions. The essential additional estimate keeps track of the
spatial support of the Fourier modes in a thin annulus. For
$R=2^n$ and relative width $\delta$, it bounds their squared family
cost by $O(n/\delta)$ and gives a rotational comparison error
\[
 O_\beta\!\left(\frac{1}{\delta n}\right)
 +O_\beta\!\left(\frac{1}{\delta^6 2^n}\right).
\]
This permits $\delta\asymp_\beta n^{-1}$ and hence the increasing number
of contractions per dyadic band. To our knowledge, this specific
width-dependent estimate and its application to log-square quadrupolar
decay in the $\RP^2$ model have not appeared in the literature we have
identified. Fourier orthogonality and the matrix-norm inequalities used
to obtain the estimate are standard.

The proof of Theorem~\ref{thm:RP3} follows that of Theorems~\ref{thm:main} and \ref{thm:highT},
with one change in the averaging step: the twelve tetrahedral rotations
are replaced by left multiplication by the quaternion group
$Q_8=\{\pm1,\pm i,\pm j,\pm k\}$. For any unit $\mathbbm{q}$, the four
vectors $\mathbbm{q},i\mathbbm{q},j\mathbbm{q},k\mathbbm{q}$ form an
orthonormal basis of $\R^4$, so the average of $Q_4$ over the orbit
vanishes, which is the analogue of \eqref{eq:Qfacts}. The rotational
estimate uses left multiplication by $i/2$, $j/2$ and $k/2$ as
noncommuting generators. The high-temperature bound has prefactor $1$
rather than $8/9$ because the single-site bound is $3/4$ rather than
$2/3$. The constants $C^{(4)}_\beta$ and $c^{(4)}_\beta$ are those of the
$\RP^3$ formalization and are not claimed to coincide with $C_\beta$
and $c_\beta$.

Through the double cover $S^3\to\SO$, $\mathbbm{q}\mapsto R(\mathbbm{q})$,
one has
$(\mathbbm{q}\cdot\mathbbm{q}')^2-\tfrac14
 =\tfrac14\tr[R(\mathbbm{q})^TR(\mathbbm{q}')]$.
The model~\eqref{eq:modelRP3} is therefore the isotropic $\SO$ matrix model
with nearest-neighbor interaction $(b_{xy}/4)\tr(R_x^TR_y)$, and
Theorem~\ref{thm:RP3} bounds its order-parameter correlation
$\av{\tr(R_x^TR_y)}$. This identification uses the standard fact that the
uniform measure on $S^3$ is mapped to the Haar measure on $\SO$. That fact
is not formalized in the present release, but it is not needed for
Theorem~\ref{thm:RP3} as stated. As discussed in
Sec.~\ref{sec:historyRP2:other}, this isotropic model is not identical to
the effective theory of the triangular-lattice Heisenberg antiferromagnet,
whose stiffnesses are in general unequal.

\section{Formalization and verification record}
\label{sec:formalization}

The Lean 4 formalization accompanying this manuscript is archived at
Zenodo~\cite{RP2Lean2026}. Its source identity, mathematical correspondence
and verification record are given below.

\subsection{Source identity and mathematical correspondence}

The archived formalization~\cite{RP2Lean2026} is the file \texttt{rp2-rp3-lean-zenodo-2026-10-21.zip},
containing the \texttt{rp2-lean} version 0.4.0 sources and the additional
$\mathbb{RP}^{3}$ proof checked on 12 October 2026 (UTC).
The Lake package version remains 0.4.0; the dated revision and source
hashes below identify the extension. All 85 original $\mathbb{RP}^{2}$
modules, their public entry point and their standalone source are
retained byte-for-byte. The public entry points are \texttt{import RP2}
and \texttt{import RP3}.

Both standalone sources are generated from the corresponding modular
sources. The supplied $\mathbb{RP}^{3}$ standalone file is byte-for-byte
identical to the regenerated file. Their SHA-256 hashes are:
\begin{center}\footnotesize
\begin{tabular}{@{}ll@{}}
\toprule
Source & SHA-256, continued on the following line\\
\midrule
\texttt{RP2-Formalization.lean} &
\texttt{843e4759583ab05a60e0b5d25cd36fe7}\\
& \texttt{a3df432f0411447bdb44a15ebc26b5ad}\\[3pt]
\texttt{RP3-Formalization.lean} &
\texttt{f0b6d06374f5565842bdf30620a715f1}\\
& \texttt{e854dbe992c5eeb5f850bdac034f43d9}\\
\bottomrule
\end{tabular}
\end{center}
For each file, concatenating its two lines gives the full hash.
Dependencies remain pinned to Lean 4.34.1 and Mathlib commit
\begin{center}\footnotesize\ttfamily
d13f23b723b8a846827a245b89c10fc7d3f11612
\end{center}
Both standalone files depend only on the pinned Mathlib environment,
with no project-module imports. The $\mathbb{RP}^{2}$ file imports
\texttt{Mathlib}; the $\mathbb{RP}^{3}$ file uses 21 targeted Mathlib
imports.

The following declarations are present in both namespaces. Each name
in this table is to be prefixed by \texttt{RP2.} or \texttt{RP3.},
as appropriate.
\begin{center}\footnotesize
\renewcommand{\arraystretch}{1.18}
\begin{tabular}{@{}>{\raggedright\arraybackslash}p{.27\textwidth}>{\raggedright\arraybackslash}p{.68\textwidth}@{}}
\toprule
Mathematical statement & Lean declaration(s), after the namespace prefix\\
\midrule
All-temperature log-square bound &
\nolinkurl{log_square_quadrupolar_decay_signed}\\
High-temperature exponential bound &
\nolinkurl{elementary_exponential_quadrupolar_decay_signed}\\
Direct path-prefix bound &
\nolinkurl{nonbacktracking_quadrupolar_bound_signed}\\
Arbitrary-power and factorial bounds &
\nolinkurl{superpolynomial_quadrupolar_decay_signed}\par
\nolinkurl{factorial_quadrupolar_decay_signed}\\
Family inequality &
\nolinkurl{Quadratic.SpinSystem.rotation_family_bound}\\
Actual thin-annulus response &
\nolinkurl{Fourier.LatticeGraph.thin_annulus_response}\\
Thin-kernel comparison &
\nolinkurl{Fourier.LatticeGraph.exists_thinKernelComparison_signed}\\
Exact normalization, one replacement &
\nolinkurl{AveragingSystem.partition_replace}\\
Product bound &
\nolinkurl{AveragingSystem.correlation_bound}\\
Local-limit bounds &
\nolinkurl{log_square_quadrupolar_decay_local_limit_signed}\par
\nolinkurl{elementary_exponential_quadrupolar_decay_local_limit_signed}\\
Absolute moments &
\nolinkurl{weighted_summability_local_limit_signed}\\
\bottomrule
\end{tabular}
\end{center}

The following additional declarations are specific to the
$\mathbb{RP}^{2}$ namespace in this release; all names in this table
have the prefix \texttt{RP2.}
\begin{center}\footnotesize
\renewcommand{\arraystretch}{1.18}
\begin{tabular}{@{}>{\raggedright\arraybackslash}p{.27\textwidth}>{\raggedright\arraybackslash}p{.68\textwidth}@{}}
\toprule
Mathematical statement & Lean declaration(s), after \texttt{RP2.}\\
\midrule
Two-endpoint factorial refinement &
\nolinkurl{two_endpoint_factorial_quadrupolar_decay_signed}\\
Paper interface for the family inequality &
\nolinkurl{Paper.family_bound}\\
Tetrahedral tensor cancellation &
\nolinkurl{tetraAction_average_Q_zero}\\
Simultaneous normalization identity &
\nolinkurl{AveragingSystem.partition_replace_all}\\
Smooth Fourier series &
\nolinkurl{contDiff_quadrupolarStructureFactor_local_limit_signed}\par
\nolinkurl{iteratedFDeriv_quadrupolarStructureFactor_local_limit_signed}\\
\bottomrule
\end{tabular}
\end{center}

The $\mathbb{RP}^{3}$ development adds the following declarations, with
prefix \texttt{RP3.}
\begin{center}\footnotesize
\renewcommand{\arraystretch}{1.18}
\begin{tabular}{@{}>{\raggedright\arraybackslash}p{.27\textwidth}>{\raggedright\arraybackslash}p{.68\textwidth}@{}}
\toprule
Mathematical statement & Lean declaration(s), after \texttt{RP3.}\\
\midrule
Quaternion cancellation &
\nolinkurl{quaternionAction_sum_qPair_zero}\\
Strict positivity at distinct sites &
\nolinkurl{ModelWitnesses.exists_distinct_sites_positive_correlation}\\
Local-limit existence &
\nolinkurl{exists_signed_local_limit_of_sequence}\par
\nolinkurl{exists_signed_free_boundary_local_limit}\par
\nolinkurl{exists_free_boundary_local_limit}\\
An actual local limit satisfying the bound &
\nolinkurl{local_limit_endpoint_nonvacuous}\\
Quaternion rotation in $\mathrm{SO}(3)$ &
\nolinkurl{SO3.rotOf_mem_SO3}\\
Trace and Gibbs-integral identities &
\nolinkurl{SO3.qPair_eq_trace}\par
\nolinkurl{SO3.correlation_eq_so3}\\
Matrix-form log-square bound &
\nolinkurl{SO3.so3_log_square_decay_signed}\\
\bottomrule
\end{tabular}
\end{center}

The formal observables are exactly
$(\boldsymbol n_x\cdot\boldsymbol n_y)^2-1/3$ on $S^2$ and
$(q_x\cdot q_y)^2-1/4$ on $S^3$, respectively. Both use normalized
spherical product measure and the squared interaction in
\eqref{eq:model}, with the spin dimension changed in the second case.
Independent spin sign changes leave both the interaction and the
observable unchanged, giving the spherical realization of projective
spins. The Gibbs denominator is proved strictly positive. The formal
radius is exactly the integer lattice $\ell^\infty$ distance. The
decay constants precede the finite retained square-lattice graph,
couplings and observation sites in the quantifiers; the signed theorems
allow every coupling array satisfying $|b_e|\le\beta$.

The generic \texttt{AnnularData} interfaces remain in the package, but
the final endpoints construct their analytic and geometric data. In
both developments, the log-square theorem obtains its
\texttt{SignedThinDyadicKernelComparison} from the proved response of
the actual model. No thin-annulus estimate, factorization certificate
or decay assumption remains as a hypothesis.

The spatial argument uses the verified dyadic frequency blocks,
thin-annulus estimates and exact positive-kernel averaging. Its
target-dependent cancellation uses twelve tetrahedral rotations for
$\mathbb{RP}^{2}$ and the quaternion group
$Q_8=\{\pm1,\pm i,\pm j,\pm k\}$ for $\mathbb{RP}^{3}$.
The latter proof works directly on $S^3$, using left multiplication by
$i/2$, $j/2$ and $k/2$ as infinitesimal generators. Its rotation paths
have period $4\pi$, and the comparison constants incorporate that
period. Appendix~\ref{app:family} gives a gauge-logarithm exposition of
the fourth-order method; the Lean developments establish the required
identities by explicit mixed derivatives and polarization. The formal
certification is of the family inequalities and final Gibbs-integral
statements, rather than a line-by-line transcription of that exposition.

The $\mathbb{RP}^{3}$ argument also includes the positive
active-edge expansion and nonbacktracking path argument. For
$0\le\beta<\log(3/2)$ it proves an exponential bound with prefactor $1$
and mass
\[
m_\beta=-\log\!\left(\max\!\left\{\tfrac12,\,
                 3(1-e^{-\beta})\right\}\right)>0.
\]
The corresponding $\mathbb{RP}^{2}$ bound has prefactor $8/9$ and the
same mass. Both bounds pass to the specified free-boundary local limits.

The $\mathrm{SO}(3)$ declarations prove
\[
(q\cdot q')^2-\tfrac14
 =\tfrac14\operatorname{Tr}\!\left(R(q)^{\mathsf T}R(q')\right)
\]
and rewrite the Gibbs quotient accordingly, with trace coupling $b_e/4$.
These formal integrals still use the spherical reference measure.
The auxiliary declaration
\nolinkurl{RP3.SO3.haarInvariance_of_surjective}
proves invariance conditional on surjectivity of the quaternion rotation
map. Surjectivity and the complete Haar-measure pushforward
identification are not formalized unconditionally in this release.
Neither is a premise of the $\mathbb{RP}^{3}$ decay theorems.

\subsection{Checks and reproducibility}

The verification records for the two targets are kept distinct.
The separate $\mathbb{RP}^{2}$ audit of 9 October 2026 compiled all
86 reachable project modules, regenerated the 16,493-line standalone
file exactly and compiled it separately. Its 69 mandatory public
theorem roots passed a raw dependency audit covering declaration types
and proof/definition bodies. Their union contained 53,315 declarations,
all of which replayed successfully in an empty kernel environment.
The newly compiled $\mathbb{RP}^{2}$ standalone object also passed
kernel replay. The supplier's earlier count of 53,300 differs through
generated auxiliary declarations; the public roots and recorded source
hashes agree. The count 53,315 refers to the separate audit, not to a
new combined-target replay.

For the revised $\mathbb{RP}^{3}$ development, the verification run of
12 October 2026 (UTC) compiled all 83 reachable project modules,
regenerated the 15,635-line standalone file and compiled it separately.
All 20 public theorem roots passed the raw dependency audit. Their
union of 53,422 declarations replayed successfully in a fresh Lean
kernel environment. The $\mathbb{RP}^{3}$ standalone compilation is
recorded separately; no additional replay of that standalone object
is claimed. These dependency counts describe separate closures and
are not additive.

Both verification records used CLI-default asynchronous elaboration
with one worker and checked all nine dependency revisions and clean
checkout states against the pins. The only axioms in the audited
public closures were \texttt{propext}, \texttt{Classical.choice} and
\texttt{Quot.sound}; no unsafe or partial proof dependency was found.
The mathematical sources contain no \texttt{sorry}, \texttt{admit} or
added axiom.

With the pinned toolchain, the $\mathbb{RP}^{2}$ release check is
\begin{verbatim}
python3 scripts/verify.py --fetch-cache --standalone \
 --replay --replay-standalone --inventory
\end{verbatim}
The second line continues the same shell command. To reproduce the
$\mathbb{RP}^{3}$ compilation, semantic checks and public-closure
replay, run
\begin{verbatim}
python3 scripts/verify_rp3.py --fetch-cache --standalone --replay
\end{verbatim}
Again, this is one shell command. It checks the delivered semantic
baseline and writes new records to
\nolinkurl{verification/rp3-latest/}.
The flag \texttt{--write-semantic-baseline} is reserved for an
intentional, separately reviewed baseline update and is not part of
ordinary reproduction.

\section{Methodology and AI-use statement}
\label{sec:method}

The human author (MO) put the question to AI (ChatGPT 6 Astra), supplying the
OpenAI release (the paper in PDF and Lean formalization on GitHub)
on the exponential decay of correlations in the $O(n)$ model \cite{OpenAI2026}.
The proof was entirely generated by ChatGPT 6 Astra.
MO then asked the AI to provide a full formalization of the proof in Lean~\cite{RP2Lean2026}.
During the process, the bound was improved from the original bound of $\exp[-c_\beta\log r\log\log r]$
to the present $\exp[-c_\beta(\log r)^2]$.
The Lean formalization was then checked with ChatGPT 6 Astra in a separate account,
and also with Claude Opus 5.5.
Furthermore, it went through an automatic verification and registered at Palomar registry~\cite{RP2-Palomar-20261023}.
The compilation, dependency and kernel-replay
checks are recorded separately in Sec.~\ref{sec:formalization}.

The human author (MO) wrote the title, the abstract, and
Sections~\ref{sec:introduction}, \ref{sec:method}, and \ref{sec:personal}.
The Lean code and proof exposition were generated by AI (ChatGPT 6 Astra).
The goal setting of the heuristic arguments in Section~\ref{sec:consequences}
was given by MO;
details were generated by ChatGPT 6 Astra, and further improved
by Claude Opus 5.5. 
MO substantially edited Section~\ref{sec:consequences} and
the historical overview in Sec.~\ref{sec:historyRP2},
and decided the overall structure of the paper.
The paper was then finalized by MO after the critical review by ChatGPT 6 Astra and Claude Opus 5.5.
As the sole human author, MO is responsible for the entire paper.

\section{Personal background and outlook on \textit{vibe-proving}}
\label{sec:personal}

As a physicist, I always long for mathematical proofs of physically interesting statements. While my own attempts are usually closer to hand-waving than real mathematical proofs, I appreciate and respect mathematical theorems. A mathematically rigorous proof is usually in a different league from physicists' arguments, and quite challenging even for many statements which are regarded as well established in (non-mathematical) physics communities. Representative examples include the existence and mass gap of Yang--Mills gauge theory (one of the Millennium Prize Problems~\cite{JaffeWitten2006}), and the absence of an infinite cluster at criticality, $\theta(p_c)=0$, for nearest-neighbor percolation on $\Z^3$
(also addressed in a claimed proof in the October 6 release~\cite{OpenAIPercolation2026}; see also references therein).
Among the October 6 releases, the claimed proof of the Haldane gap for even periodic spin-one chains,
together with a companion result on the topological index of boundary-selected infinite-volume states~\cite{OpenAIHaldane2026,OpenAIHaldaneBoundary2026}, would belong to this class.
As far as I can remember, no physicist has cast doubt about the non-vanishing excitation gap in the $S=1$ Heisenberg antiferromagnetic chain since 1990 or so,
despite the initial skepticism about Haldane's conjecture~\cite{Haldane1983}.
The Haldane gap was supported by theoretical arguments and even experimental evidence~\cite{Buyers1986}, but perhaps most convincingly by numerical simulations~\cite{WhiteHuse1993}.
Modern numerical algorithms such as Density-Matrix Renormalization Group and Quantum Monte Carlo simulation can handle rather large systems of several hundred spins with considerable accuracy, and we have very precise estimates of the gap in the $S=1$ Heisenberg antiferromagnetic chain.
However convincing these evidences are, they do not substitute a mathematically rigorous proof.
The proof of the gap remained difficult, but the OpenAI October 6 release claims such a proof.
It should be noted that, the OpenAI proof was without a Lean formalization in the release catalogue consulted on 12 October 2026,
although Tasaki~\cite{TasakiHaldane2026} has examined and provided expositions of the analytic part of the OpenAI proof.

``Practical'' physicists might question the relevance of such proofs of ``physically established statements''.
I am not sure if I can convince them of the importance of such proofs per se, but
developing mathematically rigorous frameworks will be useful, as sometimes they do provide definite answers on controversial issues.
One prominent example was the random-field Ising model. There were two quite different arguments, each of them appealing and convincing in its own way, on the random-field Ising model (RFIM). One predicted long-range order for weak disorder in three dimensions, while dimensional-reduction arguments suggested its absence~\cite{ImryMa1975,ParisiSourlas1979}.
After a long controversy, rigorous proofs established long-range order at weak disorder, first at zero temperature and then at sufficiently low positive temperatures~\cite{Imbrie1985,BricmontKupiainen1988}.
Even if such occurrences might be relatively rare, it would be important to develop mathematical frameworks for decisive applications
(besides their own aesthetic value).

Among many physics-related claims in the October 6 release, there is a statement on the exponential decay of correlation function
in the O($n$) model in 2 dimensions for $n \geq 3$, with a formalization in Lean~\cite{OpenAI2026}.
This result has perhaps attracted less attention than other claims so far.
Actually, the statement is what most physicists expected from the perturbative Renormalization-Group calculation
(which is often a textbook example of ``asymptotic freedom'') and might be classified as ``physically established statements''.
Nevertheless, it was not without controversy and there were arguments for massless phases
(with power-law correlation functions) by reputable scientists~\cite{Patrascioiu1995}.
The Lean-verified result settles the question of algebraic spin decay in the O($n$) model in 2 dimensions for $n \geq 3$ negatively.
See Appendix~\ref{app:historyO3} for more details.

Given this background, and also strong (commercial) AI models available to me, I tried \textit{vibe-proving} by myself. I just asked AI to prove a similar statement for a variant of the O($n)$ model, $\RP^2$ model in 2 dimensions.
It would be fair to say that the $\RP^2$ model has been more controversial than the O($n$) model
because of the presence of topological $\Z_2$ vortices, as discussed in earlier sections.
Without any scientific or mathematical insight supplied by me, AI did give a proof of a nontrivial bound on the correlation functions in the $\RP^2$ model. It is weaker than the exponential bound for the standard O($n$) model,
but it is strong enough to rule out  critical behavior with power-law correlations of the quadrupolar order parameter.
I believe that this is quite interesting also from a physics perspective. (It might be also interesting from a mathematical perspective, but I cannot judge how much the proof is different from the original proof on the O($n$) model in the October 6 release).
On the other hand, as a non-mathematician (and non-mathematical physicist), I am not qualified to verify the proof (at least in a short period of time). I do not think it appropriate to release unverified, possible ``AI slop'' to the general public. However, AI also managed to give a full formalization of its proof in Lean. Since it has also undergone several checks, such as the Lean checks described in Sec.~\ref{sec:formalization} and an independent check by another AI,
I suppose it is reliable, at least to the extent that warrants publication in some form.
In the old days, people tended to check outputs of electronic calculators or computers by manual calculation or with abacus,
but perhaps nobody does so any more.
These days, we also routinely rely on compilers and interpreters without manually rechecking every operation.
Only in rare cases of highly suspicious bugs identified by experts, the correctness of the compiler or interpreter is seriously examined.
Mathematical proofs might follow a similar path.

It should be, however, noted that, there is a potential pitfall in Lean formalization besides possible bugs:
possible discrepancies between the natural language statements and formalized version~\cite{Bastounis2026}.
In this paper, the correspondence between the stated theorem and the formalized version was
checked independently by both ChatGPT 6 Astra and Claude Opus 5.5.
Furthermore, the Lean formalization was automatically verified and registered at Palomar registry~\cite{RP2-Palomar-20261023}.
With these checks, we believe that the statements in this paper and the Lean formalization are consistent although
we do not claim that with absolute certainty. The proof exposition in the appendices of this paper, on the other hand, might not
be completely faithful to the Lean formalization.

Even if we accept that the formalization and thus the proof is correct, one might think that such a result which can be produced by AI
without key inputs from humans, is no longer worth publishing.
I could agree with that as a general statement. However, this is a problem which I cared about for a long time.
I was fascinated by the problem when I was a student (more than 30 years ago),
and I have occasionally thought about it since then.
In the end, I couldn't contribute much to the problem. Only recently,
I got lucky to have an excellent student who carried out
an impressive (to me) numerical study based on the tensor-network method~\cite{UedaOshikawa2022}.
While it was suggestive, it was not quite conclusive and
the problem remained as an important open problem for me.

So I decided to (ask AI to) write up the result. As I wrote above, I believe that the result is interesting from a physics perspective.
I would say, this is certainly more interesting than the average of my own papers published so far.
Also, this may be of interest as an attempt of \textit{vibe-proving} in October 2026.
Since the proof is entirely by AI,
I certainly do not claim any credit on the proof itself.
I just prompted AI with the questions (although, one might argue that that will be the important task for humans,
at least until AI starts asking better questions than we do).

It is also unclear if the present finding (mostly by AI) is appropriate for publication in the
format of a traditional scientific paper.
Perhaps in the (near) future, a new platform for scientific/mathematical
results will be developed, for more efficient and rapid dissemination of knowledge.
However, at this point the best channel is not quite clear.
As an experiment, I decided to ``write up'' the result with background and discussion
in a format following traditional scientific papers.
We will see how the scientific communications will evolve in the new era of AI-accelerated research.

As we can see from several examples discussed above and in this paper,
there is a huge ``gap'' between ``physically established statements'' (with varying degrees of certainty) and mathematical theorems.
Now we can expect to fill the ``gap'' at an unprecedented pace with the increasingly powerful AI;
occasionally we will be surprised by rigorous refutations of widely believed statements in physics or
new discoveries.
Hopefully it will bring new developments in mathematics, which will then bring unprecedented developments in physics and other sciences. 

While this is an exciting time, I also share deep concerns.
We are doing science (or mathematics) driven by curiosity and desire to find out truth.
In this respect, AI is certainly accelerating our progress in knowing truth.
At the same time, the process of discovery has been a profound joy and has had some aspects of art.
I actually find it exhilarating --- sometimes even addictive --- to elucidate a difficult problem with AI,
at a pace unimaginable without AI, at least for now as I can take the initiative of the process.
However, it is also undeniable that we are losing at least some part of the joy and art of science,
by replacing human efforts by AI. It can possibly lead to a loss of the community of
curious humans in the future.
At this point, we physicists can supply AI with questions or
conjectures based on our physical intuitions or experimental observations, to make substantial progress.
However, what will be the role of humans, if any, in the future when AI becomes much more powerful?

This is not just a problem for scientists or mathematicians; in increasingly many areas of human endeavor,
we are facing a similar problem. Throughout the history of human civilization, there have been
many cases of displacement of human labor by machines or technologies; often those were brutal
for people affected. It could be argued that when some jobs or roles are displaced,
the technology has created new jobs or roles for humans. However, it is not clear if that
will also be the case in the future with AI, and even if it is, the pace of change
might be too fast for humans to adapt.
It would be certainly great if AI will help us eliminate diseases and achieve longevity.
But what would be the point of longevity, \emph{if} the \emph{purpose of life} or \emph{joie de vivre}
is taken away from the majority of humans, not just from scientists?
I cannot exclude the possibility that ``modern Luddites'' are right and we should stop developing AI,
although I do not take that position for now.
In the case we agree to stop or abolish AI, how can we enforce that globally?
Where would be the boundary between ``AI'' (which should be stopped) and
non-AI technologies (which are to be developed further)?
In any case, these are very important issues which should be discussed in the society at large,
not just among scientists or mathematicians.

\section*{Code availability}
The complete Lean 4 source, the standalone files whose SHA-256 hashes are
listed in Sec.~\ref{sec:formalization}, the verification scripts and the
audit records are available at Zenodo~\cite{RP2Lean2026}.

\section*{Acknowledgments}

MO thanks Akihiro Tanaka for earlier discussions on the $\RP^2$ model,
and Atsushi Ueda for the related collaboration in Ref.~\cite{UedaOshikawa2022}.
This work is supported in part by
JSPS KAKENHI Grant No. JP23K25791.
Parts of this work were conducted using
ChatGPT for Academic Researchers Program provided by OpenAI
and Claude Max provided by Anthropic, as well as
a paid ChatGPT Pro account.

\appendix

\section{Controversies on the two-dimensional classical Heisenberg (O(3)) model}
\label{app:historyO3}

As discussed in the main text, the standard picture of
the two-dimensional classical Heisenberg model~\eqref{eq:model_O3}
was based on the purported monotonic RG flow of the coupling
(the effective temperature) to infinity~\cite{Polyakov1975,Brezin1976}.
However, the perturbative RG analysis alone does not
establish such an RG flow.

Patrascioiu and Seiler challenged the standard picture. Their
``super-instanton'' argument questioned the uniformity of perturbation
theory when the low-temperature and infinite-volume limits are
interchanged \cite{Patrascioiu1995}. Niedermayer, Niedermaier and Weisz
responded by analyzing the special boundary conditions involved, followed
by a rejoinder from Patrascioiu and Seiler
\cite{Niedermayer1997,Patrascioiu1998Arguments}. A second line of argument
related a possible massless phase to percolation properties of spin
projections, with constrained models and universality assumptions playing
an essential role \cite{Patrascioiu2002Percolation}.
While these challenges were substantive, they did not provide
unconditional proofs of a massless phase in the model.

The controversy continued after the 1999 open-problem statement~\cite{Princeton1999}.
All\`{e}s, Alonso, Criado and Pepe numerically investigated the equatorial
clusters used in the percolation argument \cite{Alles1999Percolation}.
Patrascioiu and Seiler objected that percolation at $\beta=J/T=2$, where
both sides accepted a finite correlation length, did not settle the
behavior at lower temperatures \cite{Patrascioiu2000Comment}. The reply
emphasized the restricted temperature range of the numerical conclusion
and questioned whether percolation of the relevant cluster implied a
massless induced rotator model \cite{Alles2000Reply}. A related exchange
concerned the apparent agreement between discrete icosahedral spin models
and $O(3)$ \cite{Patrascioiu2001Discrete,Hasenfratz2001}:
Caracciolo, Montanari and Pelissetto argued that a relevant discrete
anisotropy could remain inconspicuous until very large scales
\cite{Caracciolo2001}. Agreement over accessible distances therefore did
not settle equality of continuum limits. Seiler's review gives an account
of the dissenting views \cite{Seiler2003}.

The numerical difficulty persisted. Finite-size extrapolation had already
reached inferred correlation lengths of order $10^5$ in 1995
\cite{Caracciolo1995}. Much later, the tensor-network study of Schmoll
et al. found data compatible with the conventional picture while leaving
a BKT-like alternative open \cite{Schmoll2021}.
More recent tensor-network calculations found evidence for crossover toward the
zero-temperature limit \cite{UedaOshikawa2022,Burgelman2023}. The large
Monte Carlo study of Yao et al. subsequently tested and disfavored a
specific finite-temperature BKT fit \cite{Yao2025}. Such comparisons can
rule out particular scenarios, while they still leave the possibility
of massless phase under a different picture open.
On the rigorous side, Aru, Garban and Sep\'ulveda
developed the geometric framework and exhibited counterexamples to a
general inference from dominant percolating strong rotator couplings to
masslessness \cite{Aru2025}.

The recent OpenAI preprint proves exponential spin decay for the
nearest-neighbor ferromagnetic $O(n)$ model, $n\ge3$, uniformly over finite
square-lattice subgraphs with bounded nonnegative couplings and free
boundary, and in their free-boundary limits \cite{OpenAI2026}.
For $n=3$, this excludes the proposed algebraic spin phase within that
class, settling an important part of the long-standing controversy.

On the other hand,
it does not by itself establish the precise low-temperature growth
law of the correlation length or every claim about the continuum limit.
Nor should spin decay alone be conflated with a statement about all
topological observables: recent proposals based on topological-charge
fluctuations ask a distinct question \cite{TangShi2025,Tang2026}.

On the technical side,
the fourth-order rotational argument of the OpenAI proof is the starting
point of the present work.

\section{Rotation calculus and the fourth-order family inequality}
\label{app:input}
\label{app:family}

We first establish the finite-graph input from
Ref.~\cite[Section 2 and Proposition 2.2]{OpenAI2026} for the interaction
in \eqref{eq:model}. Temporarily allow any finite graph, orient each edge
once, and pin an arbitrary set $P\subset V$ to arbitrary spins in $S^2$.
Expectations below are in this normalized pinned measure. Let
$a,b,c\in\so$ be cross product with an oriented orthonormal basis; thus
$[a,b]=c$ and all three generators have operator norm one. For a real
skew-matrix edge field $A=(A_e)$, put
\begin{equation}\label{eq:twisted}
 H(A)=\sum_{e=(x,y)}b_e
 \bigl[(\boldsymbol n_x^Te^{A_e}\boldsymbol n_y)^2
       -(\boldsymbol n_x^T\boldsymbol n_y)^2\bigr],
 \qquad W(A)=\av{e^{H(A)}}.
\end{equation}
Write $X=DH(0)$ and $H_j=D^jH(0)$. Compactness justifies differentiation
under the integral. Derivatives are extended complex multilinearly;
changes of variables use real rotations only.

Each mixed derivative of $\boldsymbol n_x^Te^A\boldsymbol n_y$ in
norm-one generators has absolute value at most one. The product rule for
its square therefore gives, for $1\le j\le4$,
\begin{equation}\label{eq:local}
 H_j(v_1g_1,\ldots,v_jg_j)
 =\sum_eJ_e(g_1,\ldots,g_j)\prod_{s=1}^jv_s(e),
 \qquad |J_e(g_1,\ldots,g_j)|\le2^j\beta\le M:=16\beta.
\end{equation}
The coefficients are real. No sign condition on $b_e$ is used.
For $df(e)=f_y-f_x$ and scalar edge fields, define
\begin{equation}\label{eq:R}
 \mathcal R(h,k)=-D^2W(0)[hc,kc]
 =-\av{H_2(hc,kc)+X(hc)X(kc)}.
\end{equation}
This is a symmetric complex-bilinear form with real coefficients;
$\mathcal R(h,\overline h)$ is real, but positivity is not assumed.
The response differentiates $W$, not $\log W$. All $\ell^p$ norms use
counting measure. For an edge-indexed matrix,
$\norm T_{\HS}^2=\sum_{e,\ell}|T_{e\ell}|^2$.

\begin{lemma}[Fourth-order family inequality]\label{lem:family}
For a finite family of complex vertex functions $(f_i)$ vanishing on $P$,
set
\begin{equation}\label{eq:familydefs}
 u_i=df_i,\quad T_{e\ell}=\sum_i u_i(e)\overline{u_i(\ell)},\quad
 S(f,g)_e=\tfrac12(f_xg_y-g_xf_y),\quad
 q=\sum_i S(f_i,\overline f_i).
\end{equation}
Then, with $K(M)=6M^2+4M^{3/2}+M$,
\begin{equation}\label{eq:family}
 \mathcal R(q,\overline q)\le K(M)\norm T_{\HS}^2
 +4\beta\sum_{i,j}
 \min\{\norm{f_i}_\infty^2\norm{df_j}_2^2,
        \norm{f_j}_\infty^2\norm{df_i}_2^2\}.
\end{equation}
All constants are uniform in the graph, family, pinned set, and pinned values.
\end{lemma}

\begin{proof}
We give the three ingredients of the fourth-order calculation. First,
for real $f$ vanishing on $P$, a sitewise rotation gives
$W(t(df)g)=1$, for $g\in\{a,b,c\}$. Its second derivative, followed by
complexification, gives
\begin{equation}\label{eq:gradient}
 \av{|X((df)g)|^2}\le M\norm{df}_2^2.
\end{equation}
Similarly, the commuting gauge identity
$W((h+df)c)=W(hc)$ implies
$\mathcal R(h,df)=\mathcal R(df,h)=0$.
The sharper second-order coefficient bound in \eqref{eq:local} gives
$\mathcal R(v,\overline v)\le4\beta\norm v_2^2$.
Applying these facts to
\[
 S(f,g)+\tfrac12d(fg)=\tfrac12(f_x+f_y)\,dg,\qquad
 S(f,g)-\tfrac12d(fg)=-\tfrac12(g_x+g_y)\,df
\]
yields
\begin{equation}\label{eq:Scross}
 \mathcal R(S(f,g),\overline{S(f,g)})
 \le4\beta\min\{\norm f_\infty^2\norm{dg}_2^2,
                  \norm g_\infty^2\norm{df}_2^2\}.
\end{equation}
These are local Ward identities; see
Refs.~\cite{Driessler1979,AizenmanSimon1980}.

Second, let $F_v\in\operatorname{span}_{\R}\{a,b\}$ be small and zero
on $P$. Changing variables by $e^{-F_v}$ gives
\begin{equation}\label{eq:BCH}
 W(dF)=W(B(F)),\qquad
 B_e(F)=\log\!\left(e^{F_x}e^{F_y-F_x}e^{-F_y}\right).
\end{equation}
The local matrix logarithm satisfies
$DB(0)=0$ and $B_e(F)=\tfrac12[F_x,F_y]+O(F^3)$.
Its mixed second derivative in directions $fa,gb$ is $S(f,g)c$;
two equal generators give zero. Its mixed fourth derivative with
generator pattern $(a,b,a,b)$ also vanishes. Indeed, equivariance under
the half-turn about the first axis forces this derivative to lie in
$\R a$, and the half-turn about the second forces it to lie in $\R b$.
Their intersection is zero. Differentiating \eqref{eq:BCH} four times,
extending complex multilinearly, and summing the directions
$(u_i a,\overline u_i b,\overline u_j a,u_j b)$ therefore gives
\begin{equation}\label{eq:sumidentity}
 \mathcal R(q,\overline q)
 =\sum_{i,j}\mathcal R(S(f_i,f_j),\overline{S(f_i,f_j)})
 -\sum_{i,j}D^4W(0)[u_i a,\overline u_i b,\overline u_j a,u_j b].
\end{equation}
Only the two pairings with unlike generators survive the chain rule.

Third, bound the last sum from below. Write $X_r=X(A_r)$ and
$H_{r_1\cdots r_s}=H_s(A_{r_1},\ldots,A_{r_s})$ for the four indicated
directions. In the fifteen-term expansion of $D^4e^H$, the four terms
$(X_1X_2+H_{12})(X_3X_4+H_{34})$, after summing over $i,j$, equal
\[
 \E\left|\sum_i\{X(u_i a)X(\overline u_i b)
                  +H_2(u_i a,\overline u_i b)\}\right|^2\ge0.
\]
To control the other eleven terms, define
$t_e(\ell)=T_{e\ell}$ and $r_e^2=\sum_\ell|T_{e\ell}|^2$.
The identity
\[
 t_e=d\!\left(\sum_i u_i(e)\overline f_i\right)
\]
and \eqref{eq:gradient} imply
$\E|X(t_e g)|^2,\E|X(\overline t_e g)|^2\le Mr_e^2$.
The four remaining $H_2XX$ terms are each bounded in absolute
expectation by $M^2\sum_e r_e^2$ after the family indices are contracted.
For example, $H_{13}X_2X_4$ contracts to
$\sum_eJ_e(a,a)X(t_e b)X(\overline t_e b)$.
The two remaining $H_2H_2$ terms contract to coefficients times
$|T_{e\ell}|^2$ or $T_{e\ell}^2$, and have the same bound.
Each of the four $H_3X$ terms is bounded by
\[
 M^{3/2}\sum_e T_{ee}r_e\le M^{3/2}\norm T_{\HS}^2,
\]
because $T_{ee}\ge0$ and $\sum_eT_{ee}^2\le\norm T_{\HS}^2$.
The $H_4$ term is bounded by $M\sum_eT_{ee}^2$.
The summed fourth derivative is real, since complex conjugation
interchanges $i,j$ and permutes the four symmetric arguments. Thus
\begin{equation}\label{eq:lower}
 \sum_{i,j}D^4W(0)[u_i a,\overline u_i b,\overline u_j a,u_j b]
 \ge-K(M)\norm T_{\HS}^2.
\end{equation}
Equations~\eqref{eq:Scross}, \eqref{eq:sumidentity}, and \eqref{eq:lower}
prove the lemma.
\end{proof}

The corresponding formal statement is \texttt{RP2.Paper.family\_bound}.
The gauge logarithm above is an exposition of the fourth-order identity;
the Lean development establishes the needed identities by explicit
mixed derivatives and polarization.

\section{The Fourier family on a thin annulus}
\label{app:fourier}
\label{app:thin}

We construct the test functions for Lemma~\ref{lem:family} directly on
an annulus of variable width. The frequency blocks come from
Ref.~\cite[Section 3]{OpenAI2026}; keeping track of the spatial support
is the additional estimate needed for the log-square rate. Constants
$C$ in this appendix are universal and may increase from line to line.

\subsection{Profiles and their discrete approximation}

Choose a smooth even cutoff $\theta$, equal to one on $[-4/3,4/3]$
and zero outside $(-5/3,5/3)$, and put
$\Theta(z)=\theta(z_1)\theta(z_2)$. For $g_m=\partial_m\Theta$,
choose a smooth real cutoff $\chi_m$ supported in
$1<\norm z_\infty<2$ and equal to one on the $1/16$ neighborhood of
$\supp g_m$. Set $\phi_{m,\pm}=(\chi_m\pm g_m)/2$.
Then $\phi_{m,+}^2-\phi_{m,-}^2=g_m$, and whenever
$\norm{z-w}_\infty\le1/16$,
\begin{equation}\label{eq:profileproduct}
 \phi_{m,+}(z)\phi_{m,+}(w)-\phi_{m,-}(z)\phi_{m,-}(w)
 =\tfrac12\{g_m(z)+g_m(w)\}.
\end{equation}
Indeed, expanding the products leaves
$(\chi_m(z)g_m(w)+g_m(z)\chi_m(w))/2$; a nonzero derivative makes
the cutoff at the other point equal to one.

Let $N=2^n$, $\delta=2^{-h}$, with $1\le h$, $2h\le n$, and
$n-h\ge4$. Choose integer radii $N\le R<S\le2N$ with
$S-R=\delta N$, and put $\alpha=R/N$. Define the plateau map
\[
 P_c(t)=\begin{cases}t+c&t<-c,\\0&|t|\le c,\\t-c&t>c,\end{cases}
 \quad c=\alpha-\delta,\qquad
 T_{\alpha,\delta}(z)=\delta^{-1}(P_c(z_1),P_c(z_2)),
\]
and the rescaled profiles
\begin{equation}\label{eq:thinprofiles}
 \Theta_{\alpha,\delta}=\Theta\circ T_{\alpha,\delta},\qquad
 \psi_{m,\pm}=\delta^{-1/2}\phi_{m,\pm}\circ T_{\alpha,\delta}.
\end{equation}
The $\psi$ profiles vanish for $\norm z_\infty\le\alpha$ and for
$\norm z_\infty\ge\alpha+\delta$; $\Theta_{\alpha,\delta}$ equals
one and zero, respectively, on these regions. Since $P_c$ is
$1$-Lipschitz,
\begin{equation}\label{eq:profilebounds}
 \norm\psi_\infty\le C\delta^{-1/2},\qquad
 \operatorname{Lip}(\psi)\le C\delta^{-3/2},\qquad
 \operatorname{Lip}(\Theta_{\alpha,\delta})\le C\delta^{-1}.
\end{equation}
For a coordinate step $te_m$, $|t|\le\delta/16$, either the transformed
coordinate increment is $t/\delta$, or both transformed coordinates are
within $1/16$ of zero and the relevant derivatives vanish. Applying
\eqref{eq:profileproduct} and the one-dimensional trapezoid error bound
in the first case, and constancy in the second, gives
\begin{align}\label{eq:thinsteperror}
 \big|\Theta_{\alpha,\delta}(z+te_m)-\Theta_{\alpha,\delta}(z)
 -t\{&\psi_{m,+}(z)\psi_{m,+}(z+te_m)\notag\\[-.4ex]
     &-\psi_{m,-}(z)\psi_{m,-}(z+te_m)\}\big|
 \le C\delta^{-3}|t|^3.
\end{align}
This argument does not require the plateau map to be globally smooth.

Work on any retained nearest-neighbor graph in the closed annulus
$R\le\norm x_\infty\le S$, pinning its two boundary layers.
Orient edges in the positive coordinate directions. Define real edge fields
\begin{align}\label{eq:thinfields}
 u_e&=\Theta_{\alpha,\delta}(y/N)-\Theta_{\alpha,\delta}(x/N),
       &&e=(x,y),\notag\\
 v_e&=\frac1N\{\psi_{m,+}(x/N)\psi_{m,+}(y/N)
             -\psi_{m,-}(x/N)\psi_{m,-}(y/N)\},
       &&y=x+e_m.
\end{align}
Since $\delta N\ge16$, the preceding estimates imply
\begin{equation}\label{eq:thinfieldbounds}
 |u_e|\le C\delta^{-1}N^{-1},\qquad
 |u_e-v_e|\le C\delta^{-3}N^{-3}.
\end{equation}

\subsection{Truncated frequency blocks and their cost}

For $h\le j<n$, write $\rho_j=2^j$ and use the block
\[
 \mathcal P_j=\{p\in\Z^2:\rho_j\le p_1<2\rho_j,\ |p_2|\le\rho_j\}.
\]
Set $p^{1,+}=p$, $p^{2,+}=(p_2,p_1)$, and $p^{m,-}=-p^{m,+}$.
The four groups of modes are
\begin{equation}\label{eq:modes}
 f_{m,\pm,j,p}(x)=\rho_j^{-3/2}\psi_{m,\pm}(x/N)
       \exp\!\left(\frac{2\pi i\,p^{m,\pm}\cdot x}{16N}\right).
\end{equation}
Restrict them to the retained vertices. They vanish on both pinned layers,
regardless of the pinned spin values. Let $T,q$ be as in
\eqref{eq:familydefs}.

\begin{lemma}[Thin-annulus family estimates]\label{lem:modes}
For this family,
\begin{equation}\label{eq:qidentity}
 q=-i\lambda_{n,h}v,\qquad
 \lambda_{n,h}=
 N\sum_{j=h}^{n-1}\sum_{p\in\mathcal P_j}\rho_j^{-3}
       \sin\!\left(\frac{2\pi p_1}{16N}\right)
 \ge\frac{n-h}{2}\ge\frac n4,
\end{equation}
and
\begin{equation}\label{eq:cost}
 \norm T_{\HS}^2+
 \sum_{i,\ell}\min\{\norm{f_i}_\infty^2\norm{df_\ell}_2^2,
                     \norm{f_\ell}_\infty^2\norm{df_i}_2^2\}
 \le C\frac n\delta.
\end{equation}
The estimates survive arbitrary deletion of vertices and edges.
\end{lemma}

\begin{proof}
For one mode on a direction-$\ell$ edge,
\[
 S(f,\overline f)_e=-i\rho_j^{-3}\psi(x/N)\psi(y/N)
 \sin\!\left(\frac{2\pi p^{m,\pm}_\ell}{16N}\right).
\]
Reflection of the transverse frequency cancels the terms with $m\ne\ell$.
The two signs for $m=\ell$ give the difference in \eqref{eq:thinfields},
and coordinate interchange makes the coefficient the same for both edge
directions. Since $\rho_j\le p_1<N$,
$\sin(\pi p_1/(8N))\ge p_1/(4N)$. Each block has
$\rho_j(2\rho_j+1)$ frequencies and contributes at least $1/2$
to $\lambda_{n,h}$, proving \eqref{eq:qidentity}.

For one profile group let $U$ have columns $\rho_j^{1/2}df_{j,p}$.
Embed the annulus in the discrete torus of side $16N$. The selected
frequencies are distinct modulo $16N$, so their Fourier matrix satisfies
$F^*F=(16N)^2I$. On direction-$\ell$ edges the product rule gives
\begin{equation}\label{eq:U}
 U_\ell=P_\ell M_{\Delta_\ell\psi}F\diag(\rho_j^{-1})
 +P_\ell M_{\psi^{(\ell)}}F
  \diag\!\left(\frac{e^{2\pi i p^{m,\pm}_\ell/(16N)}-1}{\rho_j}\right).
\end{equation}
Here $P_\ell$ restricts to retained edge origins,
$\Delta_\ell\psi(x)=\psi((x+e_\ell)/N)-\psi(x/N)$, and
$\psi^{(\ell)}(x)=\psi((x+e_\ell)/N)$; $M$ denotes multiplication.
By \eqref{eq:profilebounds}, $|p^{m,\pm}_\ell|\le2\rho_j$, and
$\rho_j^{-1}\le2^{-h}=\delta$,
\begin{equation}\label{eq:thinoperator}
 \norm U_\op\le C\delta^{-1/2},\qquad
 |U_{e,p}|\le C\delta^{-1/2}/N.
\end{equation}

Every nonzero row is incident to a site in the strict annulus.
There are at most $64\delta N^2$ such retained edges: the closed shell
has $4(S+R)(S-R+1)\le32N(S-R)$ sites, and there are two positive
edge directions. Thus each column has the stronger bound
\begin{equation}\label{eq:thincolumn}
 \sum_e|U_{e,p}|^2\le C.
\end{equation}
For $D=\diag(\rho_j^{-1})$, the matrix of this group is $UDU^*$, and
the Hilbert--Schmidt ideal inequality gives
\begin{equation}\label{eq:weightedcolumns}
 \norm{UDU^*}_{\HS}^2
 \le\norm U_\op^2\sum_p|D_{pp}|^2\norm{U_{\cdot,p}}_2^2.
\end{equation}
Each block contributes at most $3$ to $\sum_p|D_{pp}|^2$.
Equations~\eqref{eq:thinoperator}--\eqref{eq:weightedcolumns}, followed
by the sum over the four groups, give $\norm T_{\HS}^2\le Cn/\delta$.
Retaining \eqref{eq:thincolumn}, rather than bounding each column by
the operator norm, saves one inverse power of $\delta$.

For the other family term,
$\norm f_\infty^2\le C\delta^{-1}\rho_j^{-3}$ and
$\norm{df}_2^2\le C\rho_j^{-1}$.
For scales $\sigma\le\rho$, use the supremum norm at scale $\rho$ and
the gradient norm at scale $\sigma$. Summing over the
$O(\rho^2\sigma^2)$ pairs gives $C\delta^{-1}\sigma/\rho$.
The geometric sum over $\sigma\le\rho$, followed by the at most $n$
choices of $\rho$, is $Cn/\delta$. This proves \eqref{eq:cost}.
Restrictions to retained vertices and edges can only decrease these
norm bounds, and the pointwise identity for $q$ is unchanged.
\end{proof}

\section{Uniform comparison of thin-annulus kernels}
\label{app:annulus}

Use the radii and parameters from Appendix~\ref{app:fourier} and center
the annulus at any $z\in\Z^2$. Let $I,O,U$ be its retained vertices at
radii $R$, $S$, and strictly between them, respectively. Define
\begin{equation}\label{eq:kernel}
 K_{R,S}(\sigma,\tau)=
 \int_{(S^2)^U}
 \exp\!\left(\sum_{e=\{v,w\}\in E_{R,S}}
           b_e(\boldsymbol n_v\cdot\boldsymbol n_w)^2\right)
 \prod_{v\in U}\dd\omega(\boldsymbol n_v),
\end{equation}
where $E_{R,S}$ consists exactly of retained edges incident to $U$, and
$\sigma,\tau$ are arbitrary arrays on $I,O$. Nearest-neighbor geometry
places both endpoints of each such edge in the closed annulus.
Boundary-only edges are left to the intervening regions in
Appendix~\ref{app:averaging}. Empty layers are allowed.
The kernel is positive and satisfies
\begin{equation}\label{eq:simultaneous}
 K_{R,S}(g\sigma,g\tau)=K_{R,S}(\sigma,\tau)
 \qquad(g\in\SO).
\end{equation}

\begin{proposition}[Thin-annulus comparison]\label{prop:thin}
There are $E_\beta,F_\beta\ge0$, uniform in the graph, couplings, center,
and boundary arrays, such that
\begin{equation}\label{eq:thincompare}
 (1-\varepsilon_{n,h})K_{R,S}(g\sigma,\tau)
 \le K_{R,S}(\sigma,\tau),\qquad g\in\mathcal T,
 \qquad
 \varepsilon_{n,h}=E_\beta\frac{2^h}{n-h}
                    +F_\beta\frac{2^{6h}}{2^n}.
\end{equation}
The application below chooses $\varepsilon_{n,h}\le1/2$.
\end{proposition}

\begin{proof}
First rotate about the third coordinate axis and put
$Z(t)=K_{R,S}(e^{tc}\sigma,\tau)$. Apply the pinned calculus of
Appendix~\ref{app:input} at each $t$, writing its response as
$\mathcal R_t$. Translating the profiles to $z$ is understood.
The change of variables
$\boldsymbol n_v^{\rm old}
 =e^{-s\Theta_{\alpha,\delta}((v-z)/N)c}\boldsymbol n_v^{\rm new}$
gives the exact identity
\begin{equation}\label{eq:boundarygauge}
 W_t(suc)=\frac{Z(t+s)}{Z(t)},\qquad
 \mathcal R_t(u,u)=-\frac{Z''(t)}{Z(t)}.
\end{equation}
It rotates each inner pinned spin and fixes each outer pinned spin;
the pinned arrays need not be constant.

Since $v$ is real, Lemmas~\ref{lem:family} and \ref{lem:modes} yield
\begin{equation}\label{eq:approxresponse}
 \mathcal R_t(v,v)
 =\lambda_{n,h}^{-2}\mathcal R_t(q,\overline q)
 \le\frac{K_\beta}{\delta n}
\end{equation}
for a nonnegative $K_\beta$. To transfer this to $u$, put $d=u-v$.
If $\mathcal R_t(u,u)\le0$, the desired upper bound is immediate.
Otherwise \eqref{eq:R} and \eqref{eq:local}, in a box with $O(N^2)$
edges, give
\[
 \av{X_t(uc)^2}_t\le M\norm u_2^2\le C_\beta\delta^{-2},
 \qquad |X_t(dc)|\le M\norm d_1\le C_\beta\delta^{-3}/N.
\]
Ordinary Cauchy--Schwarz for these random variables bounds the difference
of the current-square terms by
$C_\beta(\delta^{-4}/N+\delta^{-6}/N^2)$.
The local quadratic terms differ by at most
$M\norm d_1\norm{u+v}_\infty$, which is bounded by
$C_\beta(\delta^{-4}/N^2+\delta^{-6}/N^4)$.
Using $0<\delta\le1$ and $N\ge1$, in either case we obtain
\begin{equation}\label{eq:thinresponse}
 -\frac{Z''(t)}{Z(t)}
 \le\frac{K_\beta}{\delta n}
     +\frac{V_\beta}{\delta^6N}=:\eta.
\end{equation}
No Cauchy--Schwarz inequality for the possibly indefinite form
$\mathcal R_t$ has been used.

The positive function $Z$ is $2\pi$-periodic. At a maximum $t_0$,
$Z'(t_0)=0$ and $Z''(t)\ge-\eta\max Z$.
Twice integrating over an interval of length at most $2\pi$ gives
\begin{equation}\label{eq:periodicbound}
 \min Z\ge(1-2\pi^2\eta)\max Z.
\end{equation}
Simultaneous rotation invariance gives the same bound for either other
coordinate axis. Every element of $\mathcal T$ has a factorization into
at most three coordinate-axis rotations. Thus, when $2\pi^2\eta\le1$,
their combined loss is at most $6\pi^2\eta$, by
$(1-s)^3\ge1-3s$. If $6\pi^2\eta\ge1$, the comparison with that
loss is automatic from positivity. Taking
$E_\beta=6\pi^2K_\beta$, $F_\beta=6\pi^2V_\beta$, and weakening
$1/n$ to $1/(n-h)$ proves \eqref{eq:thincompare}.
\end{proof}

\section{Exact rotational averaging and correlation contraction}
\label{app:averaging}

We use only annuli for which Proposition~\ref{prop:thin} gives comparison
error at most $1/2$. Products of contraction factors from nested rotation
layers have a precedent in Ref.~\cite[Section 2.1]{Ioffe2002}; the
normalization identity below makes this argument exact for the present
unnormalized kernels.

For one such kernel define
\begin{equation}\label{eq:average}
 (\mathcal A K)(\sigma,\tau)
 =\frac1{12}\sum_{g\in\mathcal T}K(g\sigma,\tau)
 =:\overline K(\sigma,\tau),\qquad
 K^{\rm rem}=2K-\overline K.
\end{equation}
The comparison gives $K\ge\overline K/2$, hence
\begin{equation}\label{eq:split}
 K=\tfrac12\overline K+\tfrac12K^{\rm rem},\qquad
 K^{\rm rem}\ge0,\qquad
 \mathcal A K^{\rm rem}=\overline K.
\end{equation}
Here $\mathcal A^2=\mathcal A$. The average is invariant under a common
$\mathcal T$ rotation of the inner boundary alone, and, by
\eqref{eq:simultaneous}, of the outer boundary alone. Both kernels in
\eqref{eq:split} retain simultaneous $\mathcal T$ invariance.
This subtraction of a common positive part is the elementary algebra of
splitting constructions \cite{Nummelin1978}; it does not require a
Markov-chain interpretation.

Place finitely many nested annuli around $x$, with disjoint closed
annular regions, each separating $x$ from $y$. Intersect all regions
with the retained graph; an annulus need not be contained in $V$.
Integrate its strict interior and assign to its kernel exactly the
edges incident to that interior, as in \eqref{eq:kernel}. All other
edges remain in the intervening regions. Nearest-neighbor geometry
prevents an interaction from jumping across a boundary layer, so this
gives an exact factorization of the Gibbs integral.

\begin{lemma}[Preservation of the partition function]\label{lem:normalization}
Replacing any collection of these annular kernels by $K^{\rm rem}$
leaves the partition function unchanged.
\end{lemma}
\begin{proof}
For one annulus integrate everything strictly inside it, including any
previously replaced kernels, and absorb the inner boundary-only
interactions into a function $F(\sigma)$ of its inner boundary array.
There are no externally pinned spins in this region. Simultaneous
$\mathcal T$ invariance of every factor and of the reference measure
gives $F(g\sigma)=F(\sigma)$.
For each fixed outer array $\tau$, a change of variables and finite
averaging therefore give
\begin{align*}
 \int F(\sigma)K^{\rm rem}(\sigma,\tau)\,\dd\omega_I(\sigma)
 &=\int F(\sigma)(\mathcal A K^{\rm rem})(\sigma,\tau)
                         \,\dd\omega_I(\sigma)\\
 &=\int F(\sigma)(\mathcal A K)(\sigma,\tau)\,\dd\omega_I(\sigma)\\
 &=\int F(\sigma)K(\sigma,\tau)\,\dd\omega_I(\sigma).
\end{align*}
The equality holds before integrating the exterior. Repeat one
replacement at a time; previously replaced kernels retain the required
simultaneous invariance. Empty layers and disconnected components do
not affect the identity.
\end{proof}

\begin{lemma}[Contraction by separating annuli]\label{lem:product}
If $m$ such annuli separate $x$ from $y$, then
\begin{equation}\label{eq:product}
 |\calC_Q^{G,b}(x,y)|\le\tfrac23\,2^{-m}.
\end{equation}
The statement includes $m=0$.
\end{lemma}
\begin{proof}
Expand the numerator of the Gibbs quotient using \eqref{eq:split}.
For a term selecting $\overline K$ at some annulus, rotate all spins
and boundary variables on its inner side by a common $g\in\mathcal T$.
The measure, the internal weights and kernels, and the selected averaged
kernel are invariant. Only the observable changes, to
$\tr(gQ(\boldsymbol n_x)g^TQ(\boldsymbol n_y))$.
Averaging over $g$ kills the term by \eqref{eq:Qfacts}.
Only the term selecting every remainder survives, with coefficient
$2^{-m}$. Its weight is nonnegative and its observable has absolute
value at most $2/3$, so its numerator is at most $(2/3)Z_{\rm rem}$
in absolute value. Lemma~\ref{lem:normalization} gives
$Z_{\rm rem}=Z_{G,b}>0$. Division proves \eqref{eq:product}.
\end{proof}

The partition identity, rather than independence of annuli, controls the
denominator. Both the kernel comparison and this averaging argument
include every director configuration, including configurations with
projective defects.

\section{Packing the thin annuli and proving the main theorem}
\label{app:thinpacking}

We now select the widths so that the number of contractions grows
quadratically in the logarithm of the separation.
Choose an integer $A\ge\max\{4,8E_\beta\}$. For each sufficiently
large integer $n$, select the unique $h=h(n)\ge0$ satisfying
\begin{equation}\label{eq:thinchoice}
 A2^h<n\le2A2^h.
\end{equation}
For $n>2A$, one has $h\ge1$ and $2h\le n$. Moreover,
\[
 E_\beta\frac{2^h}{n-h}\le\frac{2E_\beta}{A}\le\frac14,
 \qquad
 F_\beta\frac{2^{6h}}{2^n}\le F_\beta\frac{n^6}{2^n}\longrightarrow0.
\]
Choose an integer $a=a(\beta)\ge8$ so that every $n\ge a$ also has
$n-h\ge4$ and $\varepsilon_{n,h}\le1/2$.
All choices depend only on $\beta$.

In the band $[2^n,2^{n+1}]$ centered at $x$, set
\begin{equation}\label{eq:thinpacking}
 w_n=2^{n-h(n)},\quad m_n=2^{h(n)-1},\quad
 R_{n,j}=2^n+2jw_n,\quad S_{n,j}=R_{n,j}+w_n,
 \quad 0\le j<m_n.
\end{equation}
The identity $2m_nw_n=2^n$ puts all these annuli strictly below the
outer edge of the band. Their widths and intervening gaps are $w_n\ge16$;
different bands are also separated. Thus the closed annuli are disjoint,
their inner cuts are nested, and their allocated edge sets do not overlap.
Equation~\eqref{eq:thinchoice} gives
\begin{equation}\label{eq:linearshellcount}
 m_n\ge\frac{n}{4A}.
\end{equation}

Put $N_0=2^a$ and, for $r=\norm{x-y}_\infty$, define
\[
 k=\#\{j\ge0:2N_0\,2^j<r\},\qquad
 M_k=\sum_{j=0}^{k-1}m_{a+j},\qquad \gamma=\frac1{4A}.
\]
Every selected annulus separates the observation sites.
Using \eqref{eq:linearshellcount} and
$k(k-1)/2\ge(k^2-1)/4$, Lemma~\ref{lem:product} gives
\begin{equation}\label{eq:packedbound}
 |\calC_Q^{G,b}(x,y)|\le\tfrac23\,2^{-M_k},\qquad
 M_k\ge\tfrac\gamma4(k^2-1).
\end{equation}
These inequalities also hold for $k=0$, by the uniform observable bound
in \eqref{eq:Qfacts}. Set $d=\gamma\log2/4$ and $B=(2/3)e^d$;
then the right side is at most $Be^{-dk^2}$.

By the definition of $k$, $r\le2N_0\,2^k$. With
$a_0=\log(2N_0+2)$ this implies
\[
 [\log(2+r)]^2\le2a_0^2+2k^2(\log2)^2.
\]
Consequently \eqref{eq:logsquare} holds with
\begin{equation}\label{eq:logsquareconstants}
 c_\beta=\frac{d}{2(\log2)^2}=\frac1{32A\log2},\qquad
 C_\beta=B\exp(2c_\beta a_0^2).
\end{equation}
This proves Theorem~\ref{thm:main} uniformly at every distance and
completes the finite-volume argument. At $\beta=0$ we may still use the
comparison error $1/2$, so no division by a vanishing error is needed.
The displayed constants are convenient choices, not optimized physical
correlation-length estimates.

\Needspace{18\baselineskip}
\section{Local limits and consequences of the main bound}
\label{app:scales}

\begin{corollary}[Free-boundary limits and summability]\label{cor:limits}
Every specified local weak limit of the free-boundary measures
\eqref{eq:model} along an exhaustion obeys \eqref{eq:logsquare} with the
same $C_\beta,c_\beta$. For every such limit, every fixed $x\in\Z^2$,
and every $s\in\R$,
\begin{equation}\label{eq:moments}
 \sum_{y\in\Z^2}(1+\norm{x-y}_\infty)^s| \calC_Q^{G,b}(x,y)|<\infty.
\end{equation}
\end{corollary}
\begin{proof}
The two-spin observable is bounded and continuous, so its expectation
passes to the limit. Only convergence of the two-spin marginals along
the specified common sequence is needed. Spins assigned outside a finite
volume to embed its measure create no boundary interactions and do not
alter the marginal once both sites are present.
For any $p>0$, the elementary inequality
$ct^2-pt\ge-p^2/(4c)$ gives directly from the main bound
\[
| \calC_Q^{G,b}(r)|\le
 C_\beta e^{p^2/(4c_\beta)}(2+r)^{-p}.
\]
Choose $p>\max\{0,s+2\}$. Since a square-lattice shell of integer
radius $r\ge1$ has $8r$ sites, the weighted sum converges.
\end{proof}

\begin{corollary}[Smooth correlation Fourier series]\label{cor:fourier}
For a limit in Corollary~\ref{cor:limits} and fixed $x$, the series
\[
 S_x(k)=\sum_{y\in\Z^2} \calC_Q^{G,b}(x,y)e^{ik\cdot(y-x)},\qquad k\in\R^2,
\]
converges absolutely, is $C^\infty$, and may be differentiated term
by term to every finite order.
\end{corollary}
\begin{proof}
Every derivative of order $m$ has a uniform majorant given by a constant
times the summand in \eqref{eq:moments} with $s=m$. Uniform convergence
of these derivative series proves the claim.
\end{proof}

Translation invariance is unnecessary for this fixed-base-point series;
when present it gives the usual quadrupolar structure factor.
Smoothness does not imply analyticity. The local-limit statements assume
the specified convergence; this proof does not establish uniqueness, a
DLR characterization, or bounds with arbitrary pinned external boundary
conditions. The uniform pinned estimates inside the proof are used only
for annular kernels of the free-boundary model.

In particular, \eqref{eq:logsquare} rules out
$|\calC_Q(r)|\sim A r^{-\eta}(\log r)^q$ with $A>0$ and finite real
$\eta,q$, but neither a singularity in a different observable sector nor
all slower-than-exponential possibilities. A rotation-invariant energy
insertion is not annihilated by \eqref{eq:Qfacts}. The additional
correlation inequalities were used to upgrade sufficiently fast decay to
exponential decay for bilinear ferromagnetic models
\cite[Theorem 4.1 and Corollary 4.2]{AizenmanSimon1980};
see also Ref.~\cite{Simon1980}.
The corresponding inequalities have not been
established here for the projective quadrupolar observable.

\begingroup
\small
\setlength{\bibsep}{2pt}
\interlinepenalty=10000
\bibliographystyle{unsrtnat}
\bibliography{rp2-paper-references}
\endgroup

\end{document}